\documentclass{article}
\usepackage[utf8]{inputenc}
\usepackage{xypic}[all]

\usepackage{graphicx,amsfonts,mathrsfs,todonotes,xcolor}
\usepackage{enumitem}

\usepackage{authblk}

\usepackage{tikz}
\usetikzlibrary{arrows,calc}

\tikzstyle{arc}=[->, shorten <=3pt, shorten >=3pt, >=stealth, line width=1.25pt]
\tikzstyle{edge}=[shorten <=2pt, shorten >=2pt, >=stealth, line width=1.5pt]
\tikzstyle{vertex}=[circle, fill=white, draw, minimum size=6pt, inner sep=0pt,
                    outer sep=0pt]
\tikzstyle{blockQ} = [rectangle, rounded corners, minimum width=3cm, text width=5.5cm, minimum height=1cm, text centered, draw=black]
\tikzstyle{block} = [rectangle, minimum width=3cm, text width=5.5cm, minimum height=1cm, text centered, draw=black]

\usepackage{float}

\usepackage{xcolor}

\usepackage{amsmath}
\usepackage{amssymb,hyperref}
\usepackage{bm} 

\usepackage{amsthm,geometry,cite}

\usepackage[linesnumberedhidden]{algorithm2e}

\usepackage[capitalize]{cleveref}

\newtheorem{theorem}{Theorem}

\newtheorem{conjecture}{Conjecture}
\newtheorem{lemma}[theorem]{Lemma}
\newtheorem{corollary}[theorem]{Corollary}
\newtheorem{proposition}[theorem]{Proposition}
\newtheorem{question}{Question}
\newtheorem{observation}[theorem]{Observation}
\newtheorem{problem}[question]{Problem}

\theoremstyle{definition}
\newtheorem{example}[theorem]{Example}

\theoremstyle{remark}
\newtheorem{remark}[theorem]{Remark}

\newcommand{\blue}[1]{\textcolor{black}{#1}}

\DeclareMathOperator{\NP}{NP}

\DeclareMathOperator{\cP}{P}
\DeclareMathOperator{\maj}{maj}

\DeclareMathOperator{\Forb}{Forb}
\DeclareMathOperator{\CSP}{CSP}

\DeclareMathOperator{\PCSP}{PCSP}
\DeclareMathOperator{\RCSP}{RCSP}
\DeclareMathOperator{\GMSNP}{GMSNP}

\newcommand{\bA}{{\mathbb A}}
\newcommand{\bB}{{\mathbb B}}
\newcommand{\bC}{{\mathbb C}}
\newcommand{\bD}{{\mathbb D}}
\newcommand{\bF}{{\mathbb F}}
\newcommand{\bG}{{\mathbb G}}
\newcommand{\bH}{{\mathbb H}}
\newcommand{\bI}{{\mathbb I}}
\newcommand{\bK}{{\mathbb K}}
\newcommand{\bL}{{\mathbb L}}
\newcommand{\bP}{{\mathbb P}}
\newcommand{\bQ}{{\mathbb Q}}
\newcommand{\bR}{{\mathbb R}}
\newcommand{\bT}{{\mathbb T}}
\newcommand{\calC}{{\mathcal C}}
\newcommand{\calD}{{\mathcal D}}
\newcommand{\calF}{{\mathcal F}}
\newcommand{\calT}{{\mathcal T}}

\newcommand{\calS}{{\mathcal S}}
\newcommand{\calP}{{\mathcal P}}
\newcommand{\COL}[1]{{\sc ${#1}$-Colouring}}

\title{Restricted CSPs and F-free Digraph Algorithmics\thanks{Barnaby Martin has been supported by EPSRC grant EP/X03190X/1. Santiago Guzm\'an-Pro has been funded by the European Research Council (Project POCOCOP, ERC Synergy Grant 101071674). Views and opinions expressed are however those of the authors only and do not necessarily reflect those of the European Union or the European Research Council Executive Agency. Neither the European Union nor the granting authority can be held responsible for them.}}

\author[1]{Santiago Guzm\'an-Pro\thanks{santiago.guzman\_pro@tu-dresden.de}}

\author[2]{Barnaby Martin\thanks{barnabymartin@gmail.com}}

\affil[1]{Institut f\"ur Algebra, TU Dresden}
\affil[2]{Durham University, UK}

\begin{document}
\date{}

\maketitle

\begin{abstract}
In recent years, much attention has been placed on the complexity of graph homomorphism problems when the input is restricted to 
$\bP_k$-free and $\bP_k$-subgraph-free graphs. We consider the directed version of this research line, by addressing the questions
\emph{is it true that digraph homomorphism problems $\CSP(\bH)$ have a $\cP$
versus $\NP$-complete dichotomy when the input is restricted to
$\vec{\bP}_k$-free (resp.\ $\vec{\bP}_k$-subgraph-free) digraphs?}
Our main contribution in this direction shows that if $\CSP(\bH)$ is $\NP$-complete, then there is
a positive integer $N$ such that $\CSP(\bH)$ remains $\NP$-hard even for $\vec{\bP}_N$-subgraph-free
digraphs. Moreover, it remains $\NP$-hard for \emph{acyclic} $\vec{\bP}_N$-subgraph-free
digraphs, and becomes polynomial-time solvable for $\vec{\bP}_{N-1}$-subgraph-free
\emph{acyclic} digraphs.
We then verify the  questions above for digraphs on three
vertices and a family of smooth tournaments.
We prove these results by establishing a connection between $\bF$-(subgraph)-free algorithmics
and constraint satisfaction theory.
On the way, we introduce \emph{restricted CSPs}, i.e., problems of the form $\CSP(\bH)$
restricted to yes-instances of $\CSP(\bH')$ --- these were called restricted homomorphism problems
by Hell and Ne\v{s}et\v{r}il.  Another main result of this paper presents a P versus NP-complete
dichotomy for these problems. Moreover, this complexity  dichotomy is accompanied by an algebraic
dichotomy in the spirit of the finite domain CSP dichotomy.

\end{abstract}



\begin{flushright}
\emph{
As little as a few years ago, most graph theorists, while passively aware of a few classical results on graph homomorphisms, would not include homomorphisms among the topics of central interest in graph theory. We believe that this perception is changing, principally because of the usefulness of the homomorphism perspective...
At the same time, the homomorphism framework strengthens the link between graph theory and other parts of mathematics, making graph theory more attractive, and understandable, to other mathematicians.
}
Pavol Hell and Jaroslav Nesetril, \emph{Graphs and Homomorphisms}, 2004 \cite{HellNesetril}.
\end{flushright}


\section{Introduction}

\subsection*{Story of the main questions}

The Hell-Ne\v{s}et\v{r}il theorem asserts that if $\bH$ is a finite undirected graph,
then $\CSP(\bH)$ is polynomial-time solvable whenever $\bH$ is either a bipartite graph
or contains a loop, and otherwise $\CSP(\bH)$  is $\NP$-complete.
In recent years, much attention has been placed on the complexity of graph homomorphism
problems when the input is restricted to $\bF$-free and $\bF$-subgraph-free graphs, i.e., to graphs
avoiding $\bF$ as an induced subgraph, and as a subgraph, respectively.
In~\cite{goedgebeurDagstuhl2024}, the authors show that $\CSP(\bC_5)$ is polynomial-time solvable when
the input is restricted to $\bP_8$-free graphs. Moreover, they show that there are finitely many 
$\bP_8$-free obstructions to $\CSP(\bC_5)$. Also, in~\cite{bonomoC38} the authors prove that
$\CSP(\bK_3)$ becomes tractable if the input is restricted to $\bP_7$-free graphs, while Huang~\cite{huangEJC51}
proved that $\CSP(\bK_4)$ remains $\NP$-hard even for $\bP_7$-free graphs. In general, several complexity
classifications are known for \COL{\bK_n} with input restriction to $\bP_k$-free graphs
(see, e.g.,~\cite[Theorem 7]{GJPS17}), however, a complete complexity classification
of graph colouring problems with input restriction fo $\bF$-(subgraph)-free graphs remains
wide open.

The finite domain CSP dichotomy~\cite{Zhuk20} (announced independently by
Bulatov~\cite{BulatovFVConjecture} and Zhuk~\cite{ZhukFVConjecture}) asserts
that digraph colouring problems also exhibit a P versus NP-complete dichotomy.
In this paper we consider the directed version of the research line
described above, where we consider the following to be the long-term question of this variant:
Is there a $\cP$ versus $\NP$-complete dichotomy of $\CSP(\bH)$ where the input
is restricted (1) to $\bF$-free digraphs? and (2) to $\bF$-subgraph-free digraphs?%
\footnote{Note that a negative answer to any of these questions implies that, if
$\cP\neq \NP$, then there are finite digraphs $\bF$ and $\bH$ such that $\CSP(\bH)$
restricted to $\bF$-(subgraph)-free digraphs is $\NP$-intermediate --- which we believe to
be a more natural problem than the current $\NP$-intermediate problems
constructed in the literature. At this point we conjecture neither a positive nor a negative answer to the previous questions.}

The Sparse Incomparability Lemma asserts that if $\CSP(\bH)$ is $\NP$-hard,
then it remains $\NP$-hard even for high-girth digraphs. Hence,
both questions above have a positive answer whenever $\bF$ is not
an oriented forest. Motivated by the literature on graph colouring problems
restricted to $\bP_k$-free and $\bP_k$-subgraph-free digraphs, we
consider the restriction of these questions to the case when $F$ is a
directed path $\vec{\bP}_k$. 

\begin{question}\label{qst:Pk}
    Is there a $\cP$ versus $\NP$-complete dichotomy of $\CSP(\bH)$ where the input
    is restricted 
    \begin{enumerate}
        \item to $\vec{\bP}_k$-free digraphs?
        \item to $\vec{\bP}_k$-subgraph-free digraphs?
    \end{enumerate}
\end{question}

In our effort to settle Question~\ref{qst:Pk} we stumble into three more
questions which we also address in this paper. Allow us to elaborate.
Clearly, if $\CSP(\bH)$ if $\NP$-hard even for $\vec{\bP}_k$-subgraph-free digraphs, 
then $\CSP(\bH)$ is $\NP$-hard for $\vec{\bP}_k$-free digraphs. In turn, 
if $\CSP(\bH)$ restricted to $\vec{\bP}_k$-homomorphism-free digraphs is $\NP$-hard, 
then $\CSP(\bH)$ restricted to $\vec{\bP}_k$-subgraph-free digraphs.  It is well-known
that a digraph $\bD$ is $\vec{\bP}_k$-homomorphism-free if and only if $\bD$ homomorphically
maps to the transitive tournament in $k-1$ vertices $\bT\bT_{k-1}$ (see, e.g,
Observation~\ref{obs:P-TT}).
Hence, a simple way to find complexity upperbounds to the problems in 
Question~\ref{qst:Pk} is to consider the complexity of $\CSP(\bH)$ restricted to
$\CSP(\bT\bT_k)$.

The last problems have the following natural generalization:
decide $\CSP(\bH)$ restricted to input digraphs $\bD$ in $\CSP(\bH')$.
We denote this problem by $\RCSP(\bH,\bH')$. These problems have been
called restricted homomorphism problems~\cite{BHM97,brewsterDAM156}
and  it was conjectured by Hell and Ne\v{s}et\v{r}il that when $\bH$ and $\bH'$
are undirected graphs, then  $\bH'\to \bH$, or $\bH$ is bipartite, and in 
these cases $\CSP(\bH)$ restricted to $\CSP(\bH')$ is in $\cP$;
and otherwise it is $\NP$-hard. This was later confirmed by 
Brewster and Graves~\cite{brewsterDAM156}, where they actually propose a
hardness condition for a broader family of digraphs $\bH'$ (see
Theorem~\ref{thm:brewster} below). It is then natural to ask our second main question.

\begin{question}\label{qst:RCSP}
    Is there a $\cP$ versus $\NP$-hard dichotomy of problems $\RCSP(\bH,\bH')$
    parametrized by finite digraphs (structures) $\bH$ and $\bH'$?
\end{question}

Clearly, every digraph $\bD$ that admits a homomorphism to some transitive tournament must be
an acyclic digraph. A natural question of a digraph $\bH$, for which $\CSP(\bH)$ is NP-complete,
is to ask if this problem remains NP-complete on acyclic instances. Notice that
this is the same problem as $\RCSP(\bH,(\bQ,<))$. If $\bH$ is an undirected graph,
then of course this will be true (choose any total order of the vertex set and orient the edges 
according to this ordering). However, when $\bH$ is not undirected, the situation is not a priori clear.
Indeed, it is addressed for some small digraphs by Hell and Mishra in~\cite{HM14}. They prove,
for example, that $\CSP({\vec \bC_3^+})$ -- where ${\vec \bC_3^+}$ is drawn in the forthcoming
Figure~\ref{fig:three-vertices} -- does indeed remain NP-complete on acyclic inputs. The general
question is not posed in \cite{HM14}, but it is perfectly natural, and we address it here.

\begin{question}\label{qst:acyclic}
    If $\CSP(\bH)$ is $\NP$-hard for a finite graph $\bH$, does 
    $\CSP(\bH)$ remain $\NP$-hard for acyclic instances? Equivalently, is
    $\RCSP(\bH,(\bQ,<))$ $\NP$-hard whenever $\CSP(\bH)$ is $\NP$-hard?
\end{question}

Some readers might have already noticed that there is a third natural question
motivated by the previous paragraph: is there a $\cP$ versus $\NP$-complete
dichotomy of $\CSP(\bH)$ restricted to $\vec{\bP}_k$-homomorphism-free digraphs?
Or more generally: is there a $\cP$ versus $\NP$-hard dichotomy of $\CSP(\bH)$
where the input is restricted to $\calF$-homomorphism-free digraphs?
(where $\calF$ is a fixed finite set of digraphs). This question can already
be settled from results in the literature: every such problem 
can be coded into \emph{monotone monadic strict NP} (MMSNP), and Feder and 
Vardi~\cite{FederVardi} proved that if finite domain CSPs have a P versus NP-complete
dichotomy, then MMSNP exhibits the same dichotomy. However, we provide an
alternative prove of this fact in Section~\ref{sect:FO-restrictions}.

\subsection*{Story of the paper}

Our paper splits into two parts. The first part focuses on relational structures,
of which digraphs are somewhat canonical examples, while the second part focuses
on digraphs specifically. In Figure~\ref{fig:flow-of-ideas} we depict the 
flow of ideas and results of this paper.

Within the first part (Sections~\ref{sect:RCSPs}--~\ref{sect:FO-restrictions}),
we begin by introducing \emph{restricted} CSPs (RCSPs)
which have been studied as restricted homomorphism problems in \cite{BHM97,brewsterDAM156}.
In particular, we elaborate a connection with promise CSPs (PCSPs) that enables
one to view RCSPs as PCSPs. 

Our first main result shows that, for every pair of finite digraphs (structures)
$\bA$ and $\bB$, the problem $\CSP(\bB)$ with input restricted to $\CSP(\bA)$ is either
in P or NP-hard (settling Question~\ref{qst:RCSP}). Moreover, this complexity
classification is accompanied with an algebraic dichotomy
(Theorem~\ref{thm:finite-RCSP-dichotomy}) which stems from the finite domain
CSP dichotomy result. We then push the previous complexity dichotomy to finite domain CSPs
with restrictions in GMSNP (Theorem~\ref{lem:GMSNP-restrictions}). Another contribution of
this work builds again on the finite domain CSP dichotomy to present a new proof of the
complexity dichotomy for  finite domain CSP restricted to $\calF$-homomorphism-free digraphs
(structures): we avoid going via MMSNP to the infinite domain CSP setting, and stay in the 
finite domain world via Lemma~\ref{lem:xDT-T-homfree} (Section~\ref{sect:FO-restrictions}).

In the second part of the paper (Sections~\ref{sect:acyclic}--\ref{sect:single-digraph}),
we leverage results from the first part, in order to study digraph CSPs, where the ultimate
focus will be on $\bF$-free and $\bF$-subgraph-free algorithmics.

 Our second main result shows that, if $\bH$ is a digraph and $\CSP(\bH)$ is $\NP$-complete, then
 there is a positive integer $N$ such that $\CSP(\bH)$ remains $\NP$-complete
even for $\vec{\bP}_N$-subgraph-free acyclic inputs (settling Question~\ref{qst:acyclic}).
Moreover, $N$ can be chosen so that $\CSP(\bH)$ is polynomial-time solvable for
$\vec{\bP}_{N-1}$-subgraph-free acyclic inputs.
This also yields a partial answer to Question~\ref{qst:Pk}: for every digraph $\bH$
there is a positive integer $N\le 4^{|H|}$ such that Question~\ref{qst:Pk} has a positive
answer restricted to $k\ge N$. We complement this general partial answer by settling Question~\ref{qst:Pk}
for digraphs $\bH$ on three vertices (Theorems~\ref{thm:3-vertices-Pk-subgraph-free}
and~\ref{thm:Pk-3vertex-classification}), and for a family of smooth tournaments
$\bT\bC_n$ (Theorems~\ref{thm:TCn-Pk-subgraph-classification} and
Theorem~\ref{thm:TC-Pk-free-classification}). We note that eventual hardness on $\vec{\bP}_N$-subgraph-free
instances does not hold in general for $\CSP(\bH)$, if $\bH$ is an infinite digraph.
We provide a counterexample (Example~\ref{ex:infinite-acyclic}) which is otherwise well-behaved
(for example, by being $\omega$-categorical). As byproducts of our work we see that
there are finitely many minimal $\vec{\bP}_3$-obstructions to  $\CSP(\bT\bC_n)$ for each positive
integer $n$ (Theorem~\ref{thm:TCn-min-obs}), and if $\bF$ is not an oriented path, then
$\CSP(\bT\bC_n)$ (and $\CSP(\vec{\bC}_3^+)$) is $\NP$-hard even for $\bF$-subgraph-free instances
(Theorem~\ref{thm:non-paths}).

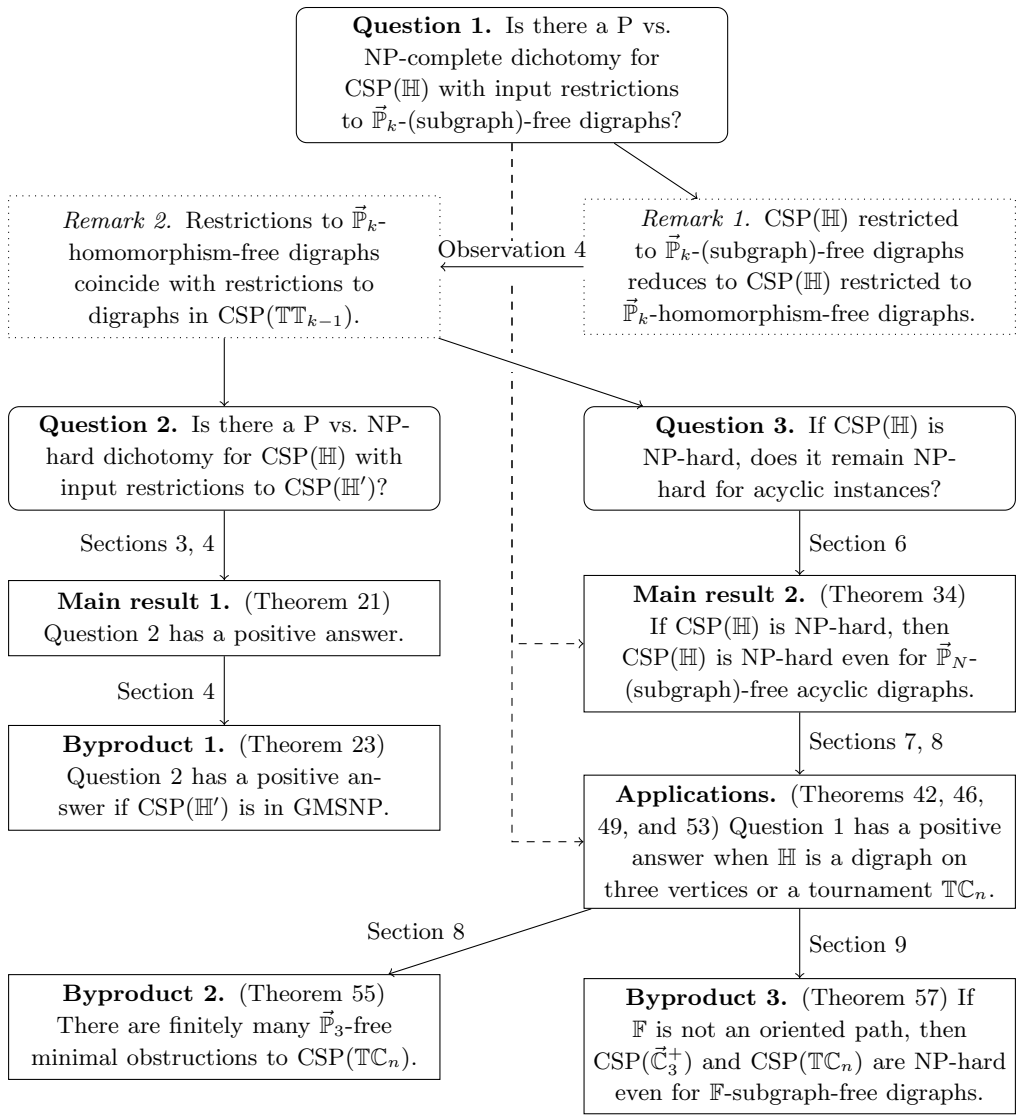
\begin{figure}[ht!]
\centering
\begin{tikzpicture}[scale = 0.85]
    
    \node[blockQ] (start) at (0,0) {\small \textbf{Question~\ref{qst:Pk}.} Is there a P vs.\ NP-complete dichotomy for $\CSP(\bH)$ with input restrictions to $\vec{\bP}_k$-(subgraph)-free digraphs?};

    \node[block, dotted] (obs1) at (4.5,-3) {\small \textit{Remark 1.} $\CSP(\bH)$ restricted to $\vec{\bP}_k$-(subgraph)-free digraphs reduces to $\CSP(\bH)$ restricted to $\vec{\bP}_k$-homomorphism-free digraphs.};

    \node[block, dotted] (obs2) at (-4.5,-3) {\small \textit{Remark 2.} Restrictions to $\vec{\bP}_k$-homomorphism-free digraphs coincide
    with restrictions to digraphs in $\CSP(\bT\bT_{k-1})$.};

    \node[blockQ] (qst1) at (4.5,-6) {\small \textbf{Question~\ref{qst:acyclic}.} If $\CSP(\bH)$ is $\NP$-hard, does
    it remain $\NP$-hard for acyclic instances?};

    \node[blockQ] (qst2) at (-4.5,-6) {\small \textbf{Question~\ref{qst:RCSP}.} Is there a P vs.\ NP-hard dichotomy for $\CSP(\bH)$ with input restrictions
    to $\CSP(\bH')$?};


    \node[block] (answ2) at (-4.5,-8.5) {\small \textbf{Main result 1.} (Theorem~\ref{thm:finite-RCSP-dichotomy})
    Question 2 has a positive answer.};
    \node[block] (bypGMSNP) at (-4.5,-11) {\small \textbf{Byproduct 1.} (Theorem~\ref{lem:GMSNP-restrictions})
    Question 2 has a positive answer if $\CSP(\bH')$ is in $\GMSNP$.};

    \node[block] (app1) at (4.5,-8.9) {\small \textbf{Main result 2.} (Theorem~\ref{thm:acyclic+bounded-paths})
    If $\CSP(\bH)$ is $\NP$-hard, then $\CSP(\bH)$ is $\NP$-hard even for 
    $\vec{\bP}_N$-(subgraph)-free acyclic digraphs.};

    \node[block] (app2) at (4.5,-12) {\small \textbf{Applications.} (Theorems~\ref{thm:3-vertices-Pk-subgraph-free},
    \ref{thm:Pk-3vertex-classification}, \ref{thm:TCn-Pk-subgraph-classification}, and~\ref{thm:TC-Pk-free-classification})
    Question 1 has a positive answer when $\bH$ is a digraph on three vertices or a tournament $\bT\bC_n$.};

     \node[block] (byp1) at (-4.5,-14.9) {\small \textbf{Byproduct 2.} (Theorem~\ref{thm:TCn-min-obs}) There
     are finitely many $\vec{\bP}_3$-free minimal obstructions to $\CSP(\bT\bC_n)$.};

     \node[block] (byp2) at (4.5,-15.2) {\small \textbf{Byproduct 3.} (Theorem~\ref{thm:non-paths})
     If $\bF$ is not an oriented path, then $\CSP(\vec{\bC}_3^+)$ and $\CSP(\bT\bC_n)$ are $\NP$-hard
     even for $\bF$-subgraph-free digraphs.};

    \draw [->, dashed] (start)  |-  (app1);
    \draw [->, dashed] (start) |- (app2);
    \draw [->] (answ2) -- node[anchor=east] {\small Section~\ref{sect:finite-domain-restrictions}} (bypGMSNP);
    \node[rectangle, fill = white, minimum height = 0.5cm, minimum width =2cm] at (0,-2.9){}; 
    \draw [->] (qst1) -- node[anchor=west] {\small Section~\ref{sect:acyclic}} (app1);
    \draw [->] (qst2) -- node[anchor=east] {\small Sections~\ref{sect:RCSPs},~\ref{sect:finite-domain-restrictions}} (answ2);
    \draw [->] (obs1) -- node[anchor=south] {\small Observation~\ref{obs:P-TT}} (obs2);
    \draw [->] (app1) -- node[anchor=west] {\small Sections~\ref{sect:3-vertices},~\ref{sect:tournaments}} (app2);
    \node[rectangle, fill = white, minimum height = 0.25cm, minimum width =2cm] at (0,-4.5){}; 
    \foreach \from/\to in {obs2/qst1, obs2/qst2, start/obs1}     
    \draw [->] (\from) to (\to);
    \draw [->] (app2) --  (byp1);
    \node[rectangle, minimum height = 0.5cm, minimum width =2cm] at (-1.5,-13.4){\small Section~\ref{sect:tournaments}};
    \draw [->] (app2) -- node[anchor = west] {\small Section~\ref{sect:single-digraph}}  (byp2);
  
\end{tikzpicture}
\caption{{\small A depiction of the flow between main results and main questions addressed in this paper.
Rectangles with rounded corners mark the main questions considered here.
Dotted squares indicate the (simple) remarks connecting the $\calF$-subgraph-free CSPs
to the theory of digraph homomorphisms. Solid rectangles with straight corners indicate the
main results of the present paper. A solid (resp.\ dashed) edge between a question and a result
indicates that the result provides an answer (resp.\ a partial answer) to the corresponding question.
Finally, edges between results represent that the result at the head is proved using tools introduced
while proving the result at the tail of the corresponding edge.}
}
\label{fig:flow-of-ideas}
\end{figure}

\section{Preliminaries}
\label{sec:preliminaries}

\subsection{Relational structures and digraphs}

For a finite relational signature $\tau=(R_1,\ldots,R_k)$, a \emph{relational ($\tau$-)structure}
$\bA$ on domain $A$ consists of $k$ relations $R_1 \subseteq A^{a_1}$, \ldots, $R_k \subseteq A^{a_k}$,
where $a_i$ is the arity of $R_i$. We denote the cardinality of $A$ as $|A|$. We tend to conflate the
relation symbol and actual relation since this will not introduce confusion. 

For some signature $\tau$, the \emph{loop} $\bL$ is the structure on one element $a$ all of whose
relations are maximally full, that is, contains one tuple $(a, \ldots, a)$ ---  the structure $\bL$
clearly depends on the signature $\tau$, but in this work $\tau$ will always be clear from context.

\emph{Directed graphs} (digraphs) $\bD$ are relational structures on the signature
$\{E\}$ where $E$ is a binary relation. A digraph is a \emph{graph} if $E$ is
symmetric, \mbox{i.e.} $(x,y) \in E$ iff $(y,x) \in E$. We will use the same blackboard 
font notation for digraphs as we do for relational structures.

Given a positive integer $k$ we denote by $\vec{\bC}_k$ the directed cycle on $k$ vertices, 
by $\vec{\bP}_k$ the directed path on $k$ vertices, by $\bT_k$ the transitive tournament
on $k$ vertices. Similarly we denote by $K_k$ the complete 
graph on $k$ vertices, and we think of it as a digraph with edges $(i,j)$ for
every $i\neq j$.

A \emph{(directed) walk} on a (directed) graph  $\bD$ is a sequence of vertices
$x_1,\dots, x_k$ such that for every $i\in[k-1]$ there is a (directed) edge
$(x_i,x_{i+1})\in E$. An \emph{oriented graph} is a digraph with no pair of symmetric edges, i.e,.
a loopless $\bK_2$-free digraph.
A \emph{tree} is an oriented digraph whose underlying graph has no cycle (equivalently, is an orientation of a traditional undirected tree).
It follows that trees are $\bK_2$-free. A \emph{forest} is a disjoint union of trees. 

\subsection{Constraint satisfaction problems}

Given a pair of digraphs (structures) $\bD$ and $\bH$ a \emph{homomorphism} $f\colon \bD\to \bH$
is a vertex mapping such that, for every $(u,v)$ that is an edge of $\bD$, the image $(f(u),f(v))$ is an edge
of $\bH$. If such a homomorphism exists, we write $\bD\to \bH$, and otherwise $\bD\not\to \bH$.
We follow notation of constraint satisfaction theory, and denote by $\CSP(\bH)$ the class
of finite digraphs such that $\bD \to \bH$. We also denote by $\CSP(\bH)$ the computational problem
of deciding if an input digraph $\bD$ belongs to $\CSP(\bH)$. 

An \emph{endomorphism} is a homomorphism $f\colon \bH\to \bH$. A digraph (structure) $\bH$ is a \emph{core}
if all of its endomorphisms are self-embeddings. If $\bH$ is finite then this is equivalent to all endomorphisms being automorphisms.

\subsection{Smooth digraphs}

A digraph $\bD$ is \emph{smooth} if it has no sources nor sinks. 
The following statement was conjectured in~\cite{bangjensenDAM26}
and proved in~\cite{BartoKozikNiven}.

\begin{theorem}[Conjecture 6.1 in~\cite{bangjensenDAM26} proved in~\cite{BartoKozikNiven}]
\label{thm:smooth}
    For every smooth digraph $\bH$ one of the following holds.
    \begin{itemize}
        \item The core of $\bH$ is a disjoint union of cycles, and in this case $\CSP(\bH)$ is
        polynomial-time solvable.
        \item Otherwise, $\CSP(\bH)$ is $\NP$-complete.
    \end{itemize}
\end{theorem}

A digraph $\bH$ is \emph{hereditarily hard} if $\CSP(\bH')$ is $\NP$-complete for every loopless digraph $\bH'$
such that $\bH\to \bH'$. Bang-Jensen, Hell, and Niven conjectured that a smooth digraph
$\bH$ is hereditarily hard whenever $\bH$ does not homomorphically map to a disjoint union of directed cycles.
Moreover, they showed that this conjecture is implied by the statement in Theorem~\ref{thm:smooth} (which was
only a conjecture at that time). 

\begin{theorem}
[Conjecture 2.5 in~\cite{bangjensenDM138} proved in~\cite{BartoKozikNiven}]
\label{thm:smooth-hereditay}
    A digraph $\bH$ is hereditarily hard whenever the digraph $R(\bH)$ obtained from $\bH$ by iteratively
    removing sources and sinks does not admit a homomorphism to a disjoint unions of directed
    cycles.
\end{theorem}

As far as we are aware, the most general result regarding hardness of restricted CSP problems
is Theorem 3 in~\cite{brewsterDAM156}.

\begin{theorem}
[Theorem 3 in \cite{brewsterDAM156}]
\label{thm:brewster}
    If $\bH$ is a hereditarily hard digraph and $\bH'$ is a finite digraph such that $\bH'\not\to \bH$
    then $\RCSP(\bH,\bH')$ is $\NP$-hard.
\end{theorem}

\subsection{Duality pairs}

A pair of digraphs (relational structure) $(\bT,\bD_\bT)$ are called
a \emph{duality pair} if for every digraph it is the case that
$\bD\to \bD_\bT$ if and only if $\bT \not\to \bD$. It was proved in~\cite{NesetrilTardif}
that for every tree $\bT$ there is a digraph $\bD_\bT$ such that $(\bT,\bD_\bT)$
is a duality pair. Moreover, they also proved that if $(\bT,\bD_\bT)$ is 
a duality pair, then $\bT$ is homomorphically equivalent to a tree.
A well-known example of a family of duality pairs is the following one.

\begin{observation}\label{obs:P-TT}
    For every positive integer $k$ the directed path $\vec{\bP}_{k+1}$
    together with the transitive tournament $\bT_k$ are duality pair.
\end{observation}

Given a set of digraphs $\calF$, we denote by $\Forb(\calF)$ the 
set of digraphs $\bD$ such that $\bF\not\to \bD$ for every $\bF\in \calF$. 
It is straightforward to observe that for any such set $\calF$, 
there is a (possibly infinite) digraph $\bD$ such that $\Forb(\calF)
= \CSP(\bD)$. 
A \emph{generalized duality} is a pair $(\calF, \calD)$ of finite
sets of digraphs such that 
\[
\Forb(\calF) = \bigcup_{\bD\in\calD} \CSP(\bD).
\]
In particular, when $\calD = \{\bD\}$ we simply write $(\calF, \bD)$.
Generalized dualities have a similar characterization to duality pairs.

\begin{theorem}[Theorems 2 and 11 in~\cite{FoniokNesetril}]\label{thm:gen-dualities}
For every finite set of forests $\calF$ there is a finite set of
digraphs $\calD$ such that $(\calD, \calF)$ is a generalized duality pair. 
Moreover, if $\calF$ is a finite set of trees, then there is a digraph
$\bD$ such that $(\calF, \bD)$ is a generalized duality.
\end{theorem}

\subsection{Large girth}

A well-known result from Erd\H{o}s~\cite{Erd:Gtp} about $k$-colourabilty states that for
every pair of positive integers $l,k$ there is a graph $G$ with girth strictly larger
than $l$ and such that $G$ does not admit a proper $k$-colouring. This result generalizes
to arbitrary relational structures, and it is known as the Sparse Incomprability Lemma~\cite{Kun}
--- in order to stay within the scope of this paper, we state it for digraphs.

Given a digraph (structure) $\bD$, the \emph{incidence graph} of $\bD$ is the undirected
bipartite graph $\bI(\bD)$ with vertex set $V\cup E$, for $v\in D$ a vertex of $D$ and $e =(x,y)\in E$ an edge of $D$. There is an (undirected) edge $(v,e)$ in $\bI(\bD)$ if and only if $v \in\{x,y\}$. The \emph{girth}
of $\bD$ is half the length of the shortest cycle in $\bI(\bD)$. Notice that if $\bD$ has a pair of symmetric
arcs, its girth is $2$, and otherwise, it is the graph theoretic girth of the underlying graph of $\bD$.

\begin{theorem}[Sparse Incomparability Lemma~\cite{Kun}]\label{thm:sparse-incomparability}
For every digraph $\bD$ and every pair of positive integers $k$ and $\ell$ there is a 
digraph $\bD'$ with the following properties:
\begin{itemize}[itemsep = 1.2pt]
    \item $\bD'\to \bD$,
    \item the girth of $\bD'$ is larger than $\ell$,  
    \item $\bD\to \bH$ if and only if $\bD'\to \bH$ for every digraph $\bH$ on
    at most $k$ vertices,
    \item $\bD'$ can be constructed in polynomial time (from $\bD$). 
\end{itemize}
\end{theorem}

\begin{corollary}\label{cor:large-girth}
    For ever finite digraph $\bH$ and every positive integer $\ell$, 
    $\CSP(\bH)$ is polynomial-time equivalent to $\CSP(\bH)$ restricted to input digraphs of girth strictly larger
    that $\ell$.
\end{corollary}

\section{Restricted constraint satisfaction problems}
\label{sect:RCSPs}

Promise problems (not to be confused with Promise CSPs) can be thought as decision problems with input restrictions. Formally~\cite{selmanIC78},
a \emph{promise problem} is a pair $(\calP,\calC)$ of decidable sets. A solution to $(\calP,\calC)$ is a decidable
set $\calS$ such that $\calS\cap \calP = \calC\cap \calP$. We say that the promise problem $(\calP,\calC)$
is polynomial-time solvable if it has a solution in P, and if every solution is NP-hard, we say
that $(\calP,\calC)$ is $\NP$-hard.

Given a pair of (possibly infinite) structures $\bA$ and $\bB$ (with the same finite signature) whose
CSPs are decidable, the \emph{restricted CSP} $\RCSP(\bA,\bB)$ is the promise problem $(\CSP(\bB),\CSP(\bA))$.
In this case, we call $(\bA,\bB)$ the \emph{template} of the restricted CSP, $\bA$ is called the \emph{domain}
and $\bB$ the \emph{restriction}. In particular, if $\bA$ is finite we say that $\RCSP(\bA,\bB)$ is
a finite domain RCSP, and if $\bB$ is finite, we say that $\RCSP(\bA,\bB)$ is an RCSP with finite restriction.

Informally, the promise problem $\RCSP(\bA,\bB)$ is $\CSP(\bA)$ where the input is promised to belong to
$\CSP(\bB)$.  For instance, $\RCSP(\bK_3,\bK_4)$ is essentially the problem of deciding whether an input
$4$-colourable graph is $3$-colourable.

Notice that for any digraph (structure) $\bA$ the problems $\CSP(\bA)$ and $\RCSP(\bA,\bL)$ are the same
problems where $\bL$ is the loop. So every decidable CSP is captured by an RCSP with finite restriction. 
One of the main results of this work shows that every RCSP with finite restriction is log-space equivalent
to a CSP. It follows from the proof that every finite domain RCSP with finite restriction is log-space equivalent
to a finite domain CSP, and thus, finite domain restricted CSPs with finite restrictions have a P versus NP-hard
dichotomy. Moreover, the reduction mentioned in this paragraph are obtained by \emph{restricted primitive positive
construction (rpp-constructions)}, which are the natural cousins of pp-constructions for CSPs~\cite{wonderland} and
for PCSPs~\cite{BartoBKO21}.

\subsection{Restricted primitive positive constructions}

Several reductions in graph algorithmics arise from gadget reductions. For instance, 
a standard way of proving that $\CSP(\bC_5)$ is $\NP$-complete can be done with the
following gadget reduction from $\CSP(\bK_5)$. On input $\bG$ to $\CSP(\bK_5)$,
consider the graph obtained from $\bG$ by replacing every edge $e:=xy$ by a path
$x,u_e,v_e,y$ (see, e.g.,~\cite{HellNesetril} where this is called an indicator construction). So in this case, the path on four vertices
is the gadget associated to this reduction). 
The algebraic approach to CSPs proposes a general framework encompassing gadget
reductions between constraint satisfaction problems. In this section we introduce
primitive positive constructions, and \emph{restricted} primitive positive constructions. 

\subsubsection*{Primitive positive constructions}

Given a pair of finite relational signature $\tau$ and $\sigma$, a \emph{primitive
positive definition} (of $\tau$ in $\sigma$)\footnote{When the signatures are clear from
context we will simply say talk about a primitive positive definition.} of dimension
$d\in \mathbb Z^+$ is a finite set $\Delta$ of primitive positive formulas $\delta_R(\bar{x})$
indexed by $\sigma$, and for each $R\in \sigma$ of arity $r$ the formula $\delta_R(\bar{x})$
has $r\cdot d$ free variables.

For every primitive positive definition of $\sigma$ in $\tau$ we associate
a mapping $\Pi_\Delta$ from $\tau$-structures to $\sigma$-structures as follows.
Given a $\tau$-structure $\bA$ the \emph{pp power} $\Pi_\Delta(\bA)$ of $\bA$
is the structure
\begin{itemize}
    \item  with domain of $\Pi_{\Delta}(\bA)$ is $A^d$, and
    \item for each $R\in \sigma$ of arity $r$ the interpretation of $R$ in $\Pi_\Delta(\bA)$
    consists of the tuples $(\bar{a}_1,\dots, \bar{a}_r)\in (A^d)^r$ such that
    $\bA\models \delta_R(\bar{a}_1,\dots, \bar{a}_r)$.
\end{itemize}
We say that a structure $\bA$ \emph{pp-constructs} a structure $\bB$ is there is a
primitive positive definition $\Delta$ such that $\Pi_\Delta(\bA) \to \bB \to  \Pi_\Delta(\bA)$, 
i.e., $\Pi_\Delta(\bA)$ is homomorphically equivalent to $\bB$. For instance,
if $\Delta$ is the $1$-dimensional primitive positive definition consisting of 
\[
\delta_E(x,y):=\exists z_1,z_2.\; E(x,z_1)\land E(z_1,z_2)\land E(z_2,y),
\]
then the pp power $\Pi_\Delta(\bC_5)$ of the $5$-cycle is the complete graph
$\bK_5$ (see Figure~\ref{fig:C5-K5}).

\begin{remark}\label{rmk:pp-power-monotone}
    It is well-known that for every fixed primitive positive definition $\Delta$, 
    the pp-power $\Pi_\Delta$ is a monotone construction with respect to the homomorphism order, 
    i.e.,  if $\bA\to \bB$, then $\Pi_\Delta(\bA)\to \Pi_\Delta(\bB)$ (this follows
    from the fact that existential positive formulas are preserved by homomorphism~\cite[Theorem 2.5.2]{Book}).
\end{remark}

\begin{figure}[ht!]
\centering
\begin{tikzpicture}

  \begin{scope}[xshift = -6cm, scale = 0.6]
    \node [vertex] (1) at (90:2) {};
    \node [vertex] (2) at (18:2) {};
    \node [vertex] (3) at (306:2) {};
    \node [vertex] (4) at (234:2) {};
    \node [vertex] (5) at (162:2) {};
    \node (L1) at (0,-2.75) {$\bC_5$};
      
    \foreach \from/\to in {1/2, 2/3, 3/4, 4/5, 5/1}     
    \draw [edge] (\from) to  (\to);
  \end{scope}

  \begin{scope}[xshift = -2.5cm, scale = 0.6]
    \node [vertex] (1) at (90:2) {};
    \node [vertex] (2) at (18:2) {};
    \node [vertex] (3) at (306:2) {};
    \node [vertex] (4) at (234:2) {};
    \node [vertex] (5) at (162:2) {};
    \node (L1) at (0,-2.75) {$\Pi_\Delta(\bC_5)\cong \bK_5$};
      
    \foreach \from/\to in {1/2, 2/3, 3/4, 4/5, 5/1, 1/3, 3/5, 5/2, 2/4, 4/1}     
    \draw [edge] (\from) to  (\to);
  \end{scope}

     \begin{scope}[xshift = 6cm, scale = 0.6]
    \node [vertex, label = above:{$\scriptstyle a_1$}] (1) at (90:1.6) {};
    \node [vertex] (11) at (145:1.3) {};
    \node [vertex] (12) at (215:1.3) {};
    \node [vertex, label = below:{$\scriptstyle b_1$}] (2) at (-90:1.6) {};
    \node [vertex] (22) at (35:1.3) {};
    \node [vertex] (21) at (-35:1.3) {};
    \node (L1) at (0,-2.75) {$\Gamma_\Delta(\bK_2)\cong \vec{\bC}_6$};
      
    \foreach \from/\to in {1/11, 11/12, 12/2}     
    \draw [arc, dashed] (\from) to  (\to);
    \foreach \from/\to in {2/21, 21/22, 22/1}     
    \draw [arc] (\from) to  (\to);
  \end{scope}

  \begin{scope}[xshift = 2.5cm, scale = 0.6]
    \node [vertex, label = above:{$\scriptstyle a$}] (1) at (90:1.6) {};
    \node [vertex, label = below:{$\scriptstyle b$}] (2) at (-90:1.6) {};
    \node (L1) at (0,-2.75) {$\bK_2$};
      
    \foreach \from/\to in {1/2, 2/1}     
    \draw [arc, dashed] (1) to [bend right] (2);
    \draw [arc] (2) to [bend right] (1);
  \end{scope}

\end{tikzpicture}
\caption{Let $\Delta$ be the primitive positive definition (of $\{E\}$ in $\{E\}$) 
where $\delta_E(x,y):=\exists z_1,z_2.\; E(x,z_1)\land E(z_1,z_2)\land E(z_2,y)$. 
On the left, we depict $\bC_5$ and its pp power $\Pi_\Delta(\bC_5)\cong \bK_5$ (an undirected
edge $xy$ represents $(x,y)$ and $(y,x)$), and
on the right, we depict $\bK_2$ and its gadget replacement $\Gamma_\Delta(\bK_2)\cong \vec{\bC}_6$ 
(dashed edges and solid edges indicate the respective edge replacements).
}
\label{fig:C5-K5}
\end{figure}
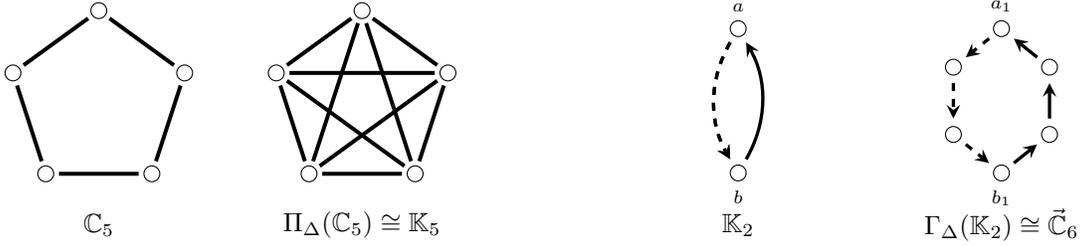

\begin{lemma}\label{lem:pp-def-compose}
    For every primitive positive definitions $\Delta_1$ of $\sigma$ in $\tau$, and
    $\Delta_3$ of $\tau$ in $\rho$, there is a primitive positive definition
    of $\Delta_3$ of $\sigma$ in $\rho$ such that for every $\rho$-structure $\bA$
    \[
    \Pi_{\Delta_3}(\bA) = \Pi_{\Delta_1}(\Pi_{\Delta_2}(\bA)).
    \]
\end{lemma}
\begin{proof}
    By substituting each relation symbol $R\in \tau$ in a formula $\delta^1_S\in\Delta_1$
    by the existential positive formula $\delta^2_R \in \Delta_2$, we obtain a set
    $\Delta_3$ of existential positive formulas $\delta^3_S$ indexed by $S\in \sigma$. It is
    not hard to observe that this primitive positive definition $\Delta_3$ satisfies
    the claim of this lemma. For further details we refer the reader to~\cite[Section 4.1.2]{Book}.
\end{proof}

Primitive positive constructions yield a canonical class of log-space reductions between
constraint satisfaction problems. 

\begin{lemma}[Corollary 3.5 in \cite{wonderland}]
\label{lem:pp-construction-reduction}
    Let $\bA$ and $\bB$ be (possibly infinite) structures with finite relational
    signature. If $\bA$ pp-constructs $\bB$, then $\CSP(\bB)$ reduces in logarithmic space
    to $\CSP(\bA)$. 
\end{lemma}

\subsubsection*{Gadget replacements}

For every primitive positive formula $\delta(\bar x)$ without equalities\footnote{If a 
primitive positive formula contains a conjunct $x = y$, we obtain an equivalent formula
$\delta'$ by deleting the conjunct $x = y$ and replacing each occurrence of the variable
$y$ by the variable $x$.} we construct a structure $\bD_\delta(\bar d)$
with a distinguished tuples of vertices $d$ as follows.
\begin{itemize}
    \item The domain $D$ of $\bD$ consists of a vertex $v_y$ for every (free or bounded) variable $y$ of $\delta$.
    \item The distinguished vertices $d_1,\dots, d_m$ are the vertices $v_{x_1},\dots, v_{x_m}$.
    \item For every relation symbol $R$ there is a tuple $\bar{v} \in R^{\bD}$ if and only if $\delta$
    contains the conjunct $R(\bar{y})$ where $\bar v$ is the tuple of vertices corresponding to the
    tuple of variables $\bar y$.
\end{itemize}
The structure $\bD(\bar d)$ is sometimes called the \emph{canonical database} of $\delta$. 
For instance, if $\delta_E$ is the formula considered above, the the canonical database of $\delta_E$
is the directed with vertices $1,2,3,4$, and the distinguished vertices are $1$ and $4$.

The canonical database $\bD_\delta(\bar{d})$ and the primitive positive formula $\delta$
are closely related: for every structure $\bA$ and a tuple $\bar a$ of elements of $A$
\[
    \bA\models \delta(\bar a) \text{ if and only if there is a homomorphism }
    f\colon \bD\to \bA \text{ such that } f(\bar d) = f(\bar a)
\]
(see, e.g.,~\cite[Proposition 1.2.5]{Book}).

Building on canonical databases, for every $d$-dimensional
primitive positive definition $\Delta$ of $\sigma$ in $\tau$
we associate a mapping $\Gamma_\Delta$ from $\sigma$-structures to $\tau$-structures as follows.
Given a $\sigma$-structure $\bB$ the \emph{gadget replacement}
$\Gamma_\phi(\bB)$ is the $\tau$-structure obtained from $\bB$ where
\begin{itemize}
    \item for each vertex $b\in B$ we introduce $d$ vertices $b_1,\dots, b_d$, and
    \item for each $R\in \sigma$ of arity $r$, and every tuple $(b^1,\dots, b^r)\in R^\bB$
    we introduce a fresh copy of $\bD_{\delta_R}(\bar{d}^1, \dots, \bar{d}^r)$ 
    and we identify each $d$-tuple $\bar{d}^i$ with $(b_1^i,\dots, b_d^i)$.
\end{itemize}
Going back to our on going example $\Delta := \{\delta_E\}$ and taking $\bB : = \bK_2$, 
the gadget replacement $\Gamma_\Delta(\bB)$ is isomorphic to the directed
$6$-cycle $\vec{\bC}_6$ (see Figure~\ref{fig:C5-K5} for a depiction).

It is straightforward to observe that, in general, for every primitive positive
definition $\Delta$, the gadget replacement $\Gamma_\Delta(\bB)$
can be constructed in logarithmic space from $\bB$. This fact and the following
observation can be used to prove Lemma~\ref{lem:pp-construction-reduction}.

\begin{observation}
[Observation 4.4 in~\cite{krokhinSIAM23}]
\label{obs:adjoint}
The following statement holds for every primitive positive formula $\phi$,
and every pair of digraphs (structures) $\bA$ and $\bA'$
\[
    \bA' \to \Pi_\Delta(\bA) \text{ if and only if } \Gamma_\Delta(\bA')\to \bA.
\]
\end{observation}

\subsubsection*{Log-space reductions and rpp-constructions}
\label{subsec:rpp}

We say that a (not necessarily finite) restricted CSP template
$(\bA,\bB)$ \emph{rpp-constructs} a restricted CSP template $(\bA',\bB')$ if
there is a primitive positive definition $\Delta$  such that
\[
(\Pi_\Delta(\bA) \times \bB') \leftrightarrow  (\bA' \times \bB') \text{ and } \bB'\to \Pi_\Delta(\bB).
\]

\begin{remark}\label{rmk:pp-and-rpp}
    Every primitive positive formula is satisfiable in some finite structure, e.g., in the loop $\bL$. 
    This implies that for every primitive positive definition of $\sigma$ in $\tau$ the
    pp-power $\Pi_\Delta(\bL_\tau)$ of the ($\tau$-)loop is homomorphically equivalent to the
    ($\sigma$-)loop $\bL_\sigma$.  In particular, this implies that the following statement are
    equivalent for every $\tau$-structure $\bA$ and every $\sigma$-structure $\bA'$.
    \begin{itemize}
        \item $\bA$ pp-constructs $\bA'$.
        \item $(\bA,\bL_\tau)$ rpp-constructs $(\bA',\bL_\sigma)$.
        \item $(\bA,\bL_\tau)$ rpp-constructs $(\bA',\bB')$ for ever structure $\bB'$.
    \end{itemize}
\end{remark}

It is well-known that pp-constructions compose, and building on this fact, it is
straightforward to observe that rpp-constructions compose. 

\begin{lemma}\label{lem:rpp-compose}
    Consider three restricted CSP templates $(\bA_1,\bB_1)$, $(\bA_2,\bB_2)$, 
    and $(\bA_3,\bB_3)$. If $(\bA_1,\bB_1)$ rpp-constructs $(\bA_2,\bA_2)$, 
    and $(\bA_2,\bB_2)$ rpp-constructs $(\bA_3,\bB_3)$, then $(\bA_1,\bB_1)$ 
    rpp-constructs $(\bA_3,\bB_3)$. 
\end{lemma}
\begin{proof}
    For $i\in\{1,2\}$, let $\Delta_i$ be a primitive positive definition witnessing that
    $(\bA_i,\bB_i)$ rpp-constructs $(\bA_{i+1},\bB_{i+1})$. Let $\Delta_3$ by a primitive
    positive definition such that $\Pi_{\Delta_3}=\Pi_{\Delta_2}\circ \Pi_{\Delta_1}$
    (Lemma~\ref{lem:pp-def-compose}). It is straightforward to observe that
    $\bB_3\to \Pi_{\Delta_3}(\bB_1)$:
    since  $\bB_3\to \Pi_{\Delta_2}(\bB_2)$, and $\bB_2\to \Pi_1(\bB_1)$, it follows (via
    Remark~\ref{rmk:pp-power-monotone}) that
    \[
        \bB_3\to \Pi_{\Delta_2}(\bB_2) \to \Pi_{\Delta_2}(\Pi_{\Delta_1}(\bB_1)) = \Pi_{\Delta_3}(\bB_1).
    \]
    It is also not hard to notice that $\Pi_{\Delta_3}(\bA_1)\times \bB_3$ is homomorphically equivalent to
    $\bA_3\times \bB_3$: by the choice of $\Delta_3$ we know that
    $\Pi_{\Delta_3}(\bA_1) = \Pi_{\Delta_2}(\Pi_{\Delta_1}(\bA_1)) = \Pi_{\Delta_2}(\bA_2)$, and since
    $\Pi_{\Delta_2}(\bA_2)\times \bB_3$ is homomorphically equivalent to $\bA_3\times \bB_3$, we
     conclude that
     \[
        (\Pi_{\Delta_3}(\bA_1) \times \bB_3) \leftrightarrow  (\bA_3 \times \bB_3) \text{ and } \bB_3\to \Pi_{\Delta_3}(\bB_1).
    \]
    Therefore, $\Delta_3$ is a primitive positive definition witnessing that
    $(\bA_1,\bB_1)$ rpp-constructs $(\bA_3,\bB_3)$.
\end{proof}

\begin{lemma}\label{lem:rpp-constructions}
    Consider two (not necessarily finite) restricted CSP templates $(\bA,\bB)$ and
    $(\bA',\bB')$. If $(\bA,\bB)$ rpp-constructs $(\bA',\bB')$, then there is a log-space
    reduction from $\RCSP(\bA',\bB')$ to $\RCSP(\bA,\bB)$.
\end{lemma}
\begin{proof}
    Let $\Delta$ be a primitive positive definition witnessing that
    $(\bA,\bB)$ rpp-constructs $(\bA',\bB')$, and let $\bC$ be an instance to $\RCSP(\bA',\bB')$.
    Since $\bC\to  \bB'$ and $\bB'\to \Pi_\Delta(\bB)$, we know that $\bC\to \Pi_\Delta(\bB)$.
    Hence, by Observation~\ref{obs:adjoint}, it follows that $\Gamma_\Delta(\bC)\to \bB$, 
    so $\Gamma_\Delta(\bC)$ is a valid instance to $\RCSP(\bA,\bB)$. Now we show that $\bC$
    is a yes-instance to $(\bA',\bB')$ if and only if $\Gamma_\Delta(\bC)$ is a yes-instance
    to $\RCSP(\bA,\bB)$. Since $\bC\to \bB$, we know that $\bC\to \bA'$ if and only if
    $\bC\to \bA'\times \bB'$, which in turn is the case if and only if 
    $\bC\to  \Pi_\Delta(\bA)\times \bB'$. Again, using our assumption that $\bC\to  \bB'$
    we see that $\bC\to  \Pi_\Delta(\bA)\times \bB'$ if and only if $\bC\to \Pi_\Delta(\bA)$. 
    Now, by Observation~\ref{obs:adjoint}, $\bC \to \Pi_\Delta(\bA)$ if and only if
    $\Gamma_\Delta(\bC)\to \bA$, so putting together all these equivalences we conclude that
    $\bC\to \bA'$ if and only if $\Gamma_\Delta(\bC)\to \bA$. Therefore, $\RCSP(\bA',\bB')$
    reduces in logarithmic space to $\RCSP(\bA,\bB)$. 
\end{proof}

\begin{example}
    Observe that the restricted $\CSP$ template $(\bC_5,\bK_3)$ rpp-constructs the template
    $(\bK_5,\bL)$ via the formula
    \[
    \delta_E(x,y):=\exists z_1,z_2.\; E(x,z_1)\land E(z_1,z_2)\land E(z_2,y).
    \]
    Indeed, we already noticed that the pp power $\Pi_\Delta(\bC_5)$ is isomorphic to $\bK_5$. 
    It is also not hard to observe that the pp power $\Pi_\Delta(\bK_5)$ is homomorphically equivalent
    to the loop $\bL$. Hence, by Lemma~\ref{lem:rpp-constructions} we conclude that $\RCSP(\bC_5,\bK_3)$
    is $\NP$-hard, i.e., deciding if a $3$-colourable graph $\bG$ admits a homomorphism to $\bC_5$
    is $\NP$-hard.
\end{example}

\section{Finite domain restrictions}
\label{sect:finite-domain-restrictions}

In this section we show that for every restricted CSP template $(\bA,\bB)$ with finite restriction,
there is a structure $\bC$ such that $\RCSP(\bA,\bB)$ and $\CSP(\bC)$ are log-space equivalent. 
Moreover, if $\bA$ is also a finite structure, then $\bC$ is a finite structure. It thus follows
from the finite domain dichotomy~\cite{BulatovFVConjecture,ZhukFVConjecture} that restricted CSPs
with finite domain and finite restriction, have a P versus NP-complete dichotomy. 

\subsection{Exponential structures}

Given a pair of $\tau$-structure $\bA$ and $\bB$ the \emph{exponential} $\bA^\bB$ is the $\tau$-structure 
\begin{itemize}
    \item with vertex set all functions $f\colon B\to A$, and
    \item for each $R\in\tau$ of arity $r$ there is an $r$-tuple $(f_1,\dots, f_r)$ belongs
    to the interpretation of $R$ in $\bA^\bB$
    if and only if $(f_1(b_1),\dots, f_r(b_r))\in R^\bA$ whenever $(b_1,\dots, b_r) \in R^\bB$.
\end{itemize}
In Figure~\ref{fig:exponential} we depict an exponential digraph construction. This image is also
found in~\cite{HellNesetril}, were the reader can also find further discussion and properties
of exponential digraphs. For this paper we state the following properties of general exponential
structures.

\begin{lemma}\label{lem:properties-exponential}
    The following statements hold for all $\tau$-structures $\bA,\bB$ and $\bC$.
    \begin{itemize}
        \item $\bA^\bL$ is isomorphic to $\bA$. 
        \item $\bC\to \bA^\bB$ if and only if $\bC\times \bB\to \bA$.
    \end{itemize}
\end{lemma}
\begin{proof}
    The first statement follows by noticing that the mapping $a\mapsto f_a$ where $f_a\colon L\to A$
    is the function mapping the unique element $l\in L$ to $a$, defines an isomorphism
    from $\bA$ to $\bA^\bL$.
    For the second statement we refer the reader to~\cite[Corollary 1.5.12]{Foniok-Thesis}.
\end{proof}

\begin{corollary}\label{cor:exponential-solution}
    For every pair of (possibly infinite) $\tau$-structures $\bA$ and $\bB$,
    the restricted CSP template $(\bA^\bB,\bL)$ rpp-constructs the template $\RCSP(\bA,\bB)$.
\end{corollary}
\begin{proof}
    Let $\Delta$ be the trivial primitive positive definition of $\tau$ in $\tau$, 
    i.e., for each $R\in\tau$ or arity $r$ let $\delta_R(x_1,\dots, x_r):=R(x_1,\dots, x_r)$.
    So for every structure $\bA$ the equality $\Pi_\Delta(\bA) = \bA$ holds.
    Clearly, $\bB\to \bL = \Pi_\Delta(\bL)$, and notice that $\bA^\bB\times \bB$ is homomorphically
    equivalent to $\bA\times \bB$: since $\bA^\bB\to \bA^\bB$, it follows from the second
    item in Lemma~\ref{lem:properties-exponential} that $\bA^\bB\times \bB \to \bA$,
    and clearly $\bA^\bB\times \bB \to \bB$, so $\bA^\bB\times \bB \to \bA\times \bB$; 
    conversely, $(\bA\times \bB) \times \bB \to \bA$, so by Lemma~\ref{lem:properties-exponential},
    $\bA\times \bB\to \bA^\bB$, and since $\bA\times \bB\to \bB$, it follows that $\bA\times \bB\to 
    \bA^\bB\times \bB$. Therefore, $\Delta$ is a primitive positive definition witnessing
    that the claim of the lemma holds.
\end{proof}

\begin{figure}[ht!]
\centering
\begin{tikzpicture}[scale = 0.7]

  \begin{scope}[xshift = -8cm, yshift = -0.35cm, scale = 0.7]
    \node [vertex, label = above:{$11$}] (1) at (90:2) {};
    \node [vertex, label = 330:{$22$}] (2) at (330:2) {};
    \node [vertex, label = 210:{$33$}] (3) at (210:2) {};
      
    \foreach \from/\to in {1/2, 2/3, 3/1}     
    \draw [arc] (\from) to  (\to);
  \end{scope}

  \begin{scope}[scale = 0.7]
    \node [vertex, label = above:{$12$}] (1) at (90:2) {};
    \node [vertex, label = 30:{$32$}] (2) at (30:2) {};
    \node [vertex, label = 330:{$31$}] (3) at (330:2) {};
    \node [vertex, label = 270:{$21$}] (4) at (270:2) {};
    \node [vertex, label = 210:{$23$}] (5) at (210:2) {};
    \node [vertex, label = 150:{$13$}] (6) at (150:2) {};
      
    \foreach \from/\to in {1/2, 2/3, 3/4, 4/5, 5/6, 6/1}     
    \draw [arc] (\from) to  (\to);
  \end{scope}

\end{tikzpicture}
\caption{The exponential $\vec{\bC}_3^{\bK_2}$ where a label $ij$ represents the
function $f\colon \{1,2\}\to\{1,2,3\}$ defined by $1\mapsto i$ and $2\mapsto j$.
}
\label{fig:exponential}
\end{figure}
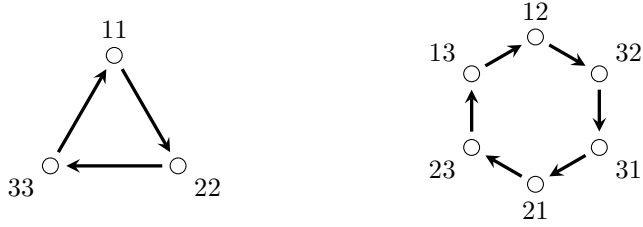

\begin{lemma}\label{lem:rpp-construction-exponential}
    Let $\tau$ be a finite relational signature, and $\bA$ be a (possible infinite) $\tau$-structure.
    If $\bB$ is a finite $\tau$-structure, then the restricted $\RCSP$ template $(\bA, \bB)$ 
    rpp-constructs the $\RCSP$ template $(\bA^\bB, \bL)$. 
\end{lemma}
\begin{proof}
    We first show that $(\bA,\bB)$ rpp-constructs $(\bA^\bB,\bL)$, and we consider the case
    of digraphs but it naturally generalizes to $\tau$-structures for finite $\tau$. Let
    $B = \{b_1,\dots, b_m\}$ and consider the following $m$-dimensional  primitive positive definition
    $\Delta := \{\delta_E\}$, where
    \[
        \delta_E(x_1,\dots, x_m, y_1,\dots, y_m):= \bigwedge_{(b_i,b_j)\in E^\bB}E(x_i, y_j).
    \]
    Notice that for any structure $\bA'$ the pp-power $\Pi_\Delta$ is isomorphic to $\bA'^\bB$. 
    Indeed, identify a tuple $(a_1,\dots, a_m)$ of $A$ with the function defined by $b_i\mapsto a_i$,
    it now follows from the definition of the exponential structure and of $\delta_E$ that this
    mapping is an isomorphism. In particular, notice that $\Pi_\Delta(\bB) = \bB^\bB$ contains the 
    loop $\bL$ (the identity function $I\colon B\to B$ is a loop on $\bB^\bB$). Hence,
    \[
        \Pi_\Delta(\bA) \cong \bA^\bB \text{ and } \bL\to \Pi_\Delta(\bB),
    \]
    i.e., $\Delta$ is a witness to the fact that $(\bA,\bB)$ rpp-constructs $(\bA^\bB,\bL)$.
\end{proof}

We discuss two applications of this lemma. The first one being that rpp-constructions
between RCSP templates with finite restrictions are captured by (standard) pp-constructions.

\begin{theorem}\label{thm:rpp-pp-constructions-finite-restriction}
    The following statements are equivalent for every pair of restricted
    $\CSP$ templates $(\bA_1,\bB_1)$ and $(\bA_2,\bB_2)$ with finite restrictions.
    \begin{itemize}
        \item  $(\bA_1,\bB_1)$ rpp-constructs $(\bA_2,\bB_2)$ and, 
        \item $\bA_1^{\bB_1}$ pp-constructs $\bA_2^{\bB_2}$.
    \end{itemize}
    In particular, $(\bA_1,\bB_1)$ and $(\bA_1^{\bB_1},\bL)$ are mutually rpp-constructible
    (and so $\RCSP(\bA_1,\bB_1)$ and $\CSP(\bA_1^{\bB_2})$ are log-space Turing-equivalent).
\end{theorem}
\begin{proof}
    Suppose that $\bA_1^{\bB_1}$ pp-constructs $\bA_2^{\bB_2}$, so by Remark~\ref{rmk:pp-and-rpp}
    the template $(\bA_1^{\bB_1},\bL)$ rpp-constructs $(\bA_2^{\bB_2},\bL)$. By
    Lemma~\ref{lem:rpp-construction-exponential}, $(\bA_1,\bB_1)$ rpp-constructs $(\bA_1^{\bB_1},\bL)$,
    and by Corollary~\ref{cor:exponential-solution} $(\bA_2^{\bB_2},\bL)$ rpp-constructs $(\bA_2,\bB_2)$.
    Hence, by composing rpp-constructions (Lemma~\ref{lem:rpp-compose}), we conclude that
    $(\bA_1,\bB_1)$ rpp-constructs $(\bA_2,\bB_2)$. The converse implication follows with similar arguments:
    $(\bA_2,\bB_2)$ rpp-constructs $(\bA_2^{\bB_2},\bL)$ (Lemma~\ref{lem:rpp-construction-exponential}), 
    and $(\bA_1^{\bB_1},\bL)$ rpp-constructs $(\bA_1,\bB_1)$ (Corollary~\ref{cor:exponential-solution});
    by composing rpp-constructions we see that $(\bA_1^{\bB_1},\bL)$ rpp-constructs $(\bA_2^{\bB_2},\bL)$,
    and it thus follows that $\bA_1^{\bB_1}$ pp-constructs $\bA_2^{\bB_2}$ (Remark~\ref{rmk:pp-and-rpp}).
\end{proof}

The second application of Lemma~\ref{lem:rpp-construction-exponential} is the following
statement which is analogous to the case of CSPs and pp-constructions (see, e.g.,~\cite[Theorem 3.2.2]{Book}).

\begin{theorem}\label{thm:rpp-construction-K3}
    The following statements are equivalent for every (possibly infinite) structure $\bA$
    and a finite structure $\bB$ with a finite signature. 
    \begin{itemize}
        \item The restricted $\CSP$ template $(\bA,\bB)$ rpp-constructs $(\bK_3,\bL)$.
        \item The restricted $\CSP$ template $(\bA,\bB)$ rpp-constructs every restricted
        $\CSP$ template $(\bA',\bB')$ where $\bA'$ is a finite structure (and $\bB'$ a possibly infinite structure).
        \item The structure $\bA^\bB$ pp-constructs $\bK_3$. 
    \end{itemize}
    If any of these equivalent statement hold, then $\RCSP(\bA,\bB)$ is $\NP$-hard.
\end{theorem}
\begin{proof}
    By the first item of Lemma~\ref{lem:properties-exponential} we know that
    $\bK_3^\bL$ is isomorphic to $\bK_3$, and so $\bA^\bB$ pp-constructs $\bK_3$
    if and only if $\bA^\bB$ pp-constructs $\bK_3^\bL$. Hence, the equivalence between
    the first and third itemized statement follows from Theorem~\ref{thm:rpp-pp-constructions-finite-restriction}.
    Also, second statement clearly implies the third one. 
    Finally, we show that the first item implies the second one. It is well-known that 
    that $\bK_3$ pp-constructs every finite structure $\bA'$ (see, e.g.,~\cite[Corollary 3.2.1]{Book}).
    So, by Remark~\ref{rmk:pp-and-rpp}
    $(\bK_3,\bL)$ rpp-constructs every finite domain $\RCSP$ template $(\bA',\bB')$
    (where $\bB'$ is a possibly infinite structure). Again, by composing rpp-constructions,
    we conclude that $(\bA,\bB)$ rpp-constructs $(\bA',\bB')$. The equivalence between the three
    itemized statement is now settled. The final statement holds via Lemma~\ref{lem:pp-construction-reduction}
    because $\RCSP(\bK_3,\bL) = \CSP(\bK_3)$ is $\NP$-hard. 
\end{proof}

\subsection{A dichotomy for finite domain RCSPs with finite restrictions}

The dichotomy for finite domain CSPs asserts that if $\bA$ is a finite structure, then
either $\bA$ pp-constructs $\bK_3$, and in this case $\CSP(\bA)$ is $\NP$-complete; or otherwise,
$\CSP(\bA)$ is polynomial-time solvable. We apply the results of this section to obtain an
analogous (actually, equivalent) statement for finite domain CSPs with finite restrictions. 

\begin{theorem}\label{thm:finite-RCSP-dichotomy}
    For every pair of finite structures $\bA,\bB$ one of the following statement hold.
    \begin{itemize}
        \item Either $(\bA,\bB)$ rpp-constructs $(\bK_3,\bL)$ (and consequently, 
        $\RCSP(\bA,\bB)$ is $\NP$-hard), or
        \item $\RCSP(\bA,\bB)$ is polynomial-time solvable. 
    \end{itemize}
\end{theorem}
\begin{proof}
   Suppose that the first itemized statement does not hold, and assume $\cP\neq\NP$
   (otherwise, $\CSP(\bA)$ is in $\cP$, and thus $\RCSP(\bA,\bB)$ is polynomial-time solvable).
   Since $(\bA,\bB)$ does not rpp-construct $(\bK_3,\bL)$, it follows from
   Theorem~\ref{thm:rpp-construction-K3} that $\bA^\bB$ does not pp-construct $\bK_3$.
   Hence, the finite domain dichotomy implies that $\CSP(\bA^\bB)$ is in $\cP$.
   By the ``in particular'' statement of Theorem~\ref{thm:rpp-pp-constructions-finite-restriction}
   we know that $\CSP(\bA^\bB)$ and $\RCSP(\bA,\bB)$ are polynomial-time equivalent,
   and we thus conclude that  $\RCSP(\bA,\bB)$ is polynomial-time solvable. 
\end{proof}

\subsection{Finite-domain restrictions up to high girth}
\label{subsect:finite-domain-up-to-high-girth}

Given a structure $\bB$ and a positive integer $\ell$ we denote by $\CSP_{>\ell}(\bB)$
the subclass of $\CSP(\bB)$ consisting of structure $\bA\in \CSP(\bB)$ with girth
strictly larger than $\ell$. The CSP of a structure $\bB$ is \emph{finite-domain up to hight girth}
if there is a positive integer $\ell$ and a finite structure $\bB'$ such that $\CSP_{>\ell}(\bB) = \CSP_{>\ell}(\bB')$.
It was proved in~\cite{guzmanGMSNP} that for every such structure $\bB$, there is a finite
structure $\bB_S$ (unique up to homomorphic equivalence) such that for every finite structure
$\bA$ there is a homomorphism $\bB\to \bA$ if and only if $\bB_S\to \bA$. We call $\bB_S$
the \emph{smallest finite factor} of $\bB$ --- notice that in particular, $\bB\to \bB_S$.
Moreover, if $\bB$ is finite domain up to high girth  and $\bB_S$ is its smallest finite factor,
then there is a positive integer $\ell$ such that
$\CSP_{>\ell}(\bB) =\CSP_{>\ell}(\bB_S)$\cite[Corollary 10]{guzmanGMSNP}.

\begin{lemma}\label{lem:up-to-high-girth-restrictions}
    Let $\bB$ be a structure such that $\CSP(\bB)$ is finite-domain up to high girth.
    If $\bB_S$ is the smallest finite factor of $\bB$, then $\RCSP(\bA,\bB)$ and
    $\RCSP(\bA,\bB_S)$ are polynomial-time equivalent.
\end{lemma}
\begin{proof}
    Since $\bB\to \bB_S$, it follows that $\RCSP(\bA,\bB_S)$ is at least as hard
    as $\RCSP(\bA,\bB)$. Fo the converse reduction, let $k := \max\{|A|,|B_S|\}$,
    and $\ell$ a positive integer such that $\CSP_{>\ell}(\bB) = \CSP_{>\ell}(\bB_S)$.
    let $\bC$ be an input to $\RCSP(\bA,\bB_S)$.
    By the Sparse Incomparability Lemma (applied to $\ell$ and $k$), there is a
    structure $\bC'$ of girth larger than $\ell$, and such that $\bC'\to \bB_S$, 
    and $\bC\to \bA$ if and only if $\bC'\to \bA$. Since $\bC'$ has girth larger
    than $\ell$ and $\CSP_{>\ell}(\bB) = \CSP_{>\ell}(\bB_S)$, it follows that
    $\bC'\to \bB$. Hence $\bC'$ is a valid input to $\RCSP(\bA,\bB)$, and since
    $\bC'$ is constructible in polynomial-time from $\bC$
    (Theorem~\ref{thm:sparse-incomparability}) and $\bC'\to \bA$ if and only if $\bC\to \bA$,
    we conclude that $\RCSP(\bA,\bB_S)$ reduces in polynomial-time to $\RCSP(\bA,\bB)$.
\end{proof}

\subsection{A  dichotomy for finite domain RCSPs with GMSNP restrictions}
\label{subsect:GMSNP}

 \emph{Forbidden pattern problems} (for graphs) are parametrized by a finite
set $\calF$ of vertex and edge coloured connected graphs, and the task is to
decide is an input graph $\bG$ admits a colouring $\bG'$ (with the same colours
used in $\calF$) such that $\bG'\in \Forb(\calF)$, i.e., there is no homomorphism
$\bF\to \bG'$ for any $\bF\in \calF$. A standard example of the problem of deciding
if an input graph $\bG$ admits a $2$-edge-colouring with no monochromatic triangles.

There is a natural fragment of existential second order logic called
\emph{guarded monotone strict NP} (GMSNP) such that CSPs expressible in
GMSNP capture forbidden pattern problems. We refer the reader
to~\cite{BarsukovThesis} for further background on GMSNP. 

Lemma 12 in~\cite{guzmanGMSNP} shows that every CSP expressible in GMSNP is
finite-domain up to high girth. The following statement is an immediate consequence
of this fact, and of Lemma~\ref{lem:up-to-high-girth-restrictions}.

\begin{theorem}\label{lem:GMSNP-restrictions}
    For every finite structure $\bA$, and every structure $\bB$ such that $\CSP(\bB)$ 
    is expressible in $\GMSNP$ one of the following statement holds.
    \begin{itemize}
        \item Either $(\bA,\bB_S)$ rpp-constructs $(\bK_3,\bL)$, where $\bB_S$ is the smallest
        finite factor of $\bB$, (and consequently, $\RCSP(\bA,\bB)$ is $\NP$-hard), or
        \item $\RCSP(\bA,\bB)$ is polynomial-time solvable. 
    \end{itemize}
\end{theorem}

\subsection{The tractability conjecture and RCSPs}

The tractability conjecture is a generalization of the Feder-Vardi conjecture to a broad class
of  ``well-behaved'' infinite structures, namely, to \emph{reducts} of
\emph{finitely bounded homogeneous} structure. A structure $\bB$ is homogeneous if for
every isomorphism $f\colon \bA \to\bA'$ between finite substructures of $\bB$, there is
an automorphism  $f'\colon \bB\to \bB$ that extends $f$,  i.e., $f'(a) = f(a)$ for every
$a\in A$. A structure $\bB$ is called finitely bounded  if there exists a finite set
${\mathcal F}$ of finite structures such that a finite structure $\bA$ embeds into $\bB$
if and only if no structure from ${\mathcal F}$ embeds into $\bA$. A reduct
of a structure $\bB$ is a structure $\bA$ obtained from $\bB$ by forgetting some relations.

\begin{conjecture}
[Conjecture 3.7.1 in~\cite{Book}]
\label{conj:tractability}
Let $\bA$ be a reduct of a finitely bounded homogeneous structure. If $\bA$ does not pp-construct
$\bK_3$, then $\CSP(\bA)$ is polynomial-time solvable. 
\end{conjecture}

This is a wide-open conjecture yielding a very active research line (e.g.,~\cite{wonderland,BKOPP,BartoPinskerDichotomy,Book,BodirskyBodorUIP,BodDalJournal,Pinsker22}).
Here we show that this conjecture is equivalent to the following conjecture for restricted CSPs
templates with finite restriction and whose domain is a reduct of a finitely bounded homogeneous structure.

\begin{conjecture}\label{conj:red-FB-homogeneous}
    Let $\bA$ be a reduct of a finitely bounded homogeneous structure, and $\bB$ be a finite
    structure.  If the restricted $\CSP$-template does not rpp-construct $(\bK_3,\bL)$, then
    $\RCSP(\bA,\bB)$ is polynomial-time solvable.
\end{conjecture}

\begin{theorem}
    Conjecture~\ref{conj:tractability} and Conjecture~\ref{conj:red-FB-homogeneous} are equivalent.
\end{theorem}
\begin{proof}
    The fact that Conjecture~\ref{conj:red-FB-homogeneous} implies Conjecture~\ref{conj:tractability}
    follows from Remark~\ref{rmk:pp-and-rpp}: if $\bA$ does not pp-construct $\bK_3$, then $(\bA,\bL)$
    does not rpp-constructs $(\bK_3,\bL)$; so, Conjecture~\ref{conj:red-FB-homogeneous} implies that
    $\RCSP(\bA,\bL)$ is polynomial-time solvable, and thus, $\CSP(\bA) = \RCSP(\bA,\bL)$ is in $\cP$. 
    For the converse implication, suppose that $(\bA,\bB)$ does not rpp-construct $(\bK_3,\bL)$. 
    By the equivalence between the second and third statement of Theorem~\ref{thm:rpp-construction-K3}, 
    it must be the case that $\bA^\bB$ does not pp-construct $\bK_3$. Moreover, in the proof
    of Lemma~\ref{lem:rpp-construction-exponential} we propose a primitive positive definition
    $\Delta$ such that $\Pi_\Delta(\bA) \cong \bA^\bB$. Since reducts of finitely bounded
    homogeneous structures are closed under pp-powers (because reducts of finitely bounded 
    homogeneous structures are closed under products and pp-definitions), it follows
    that $\bA^\bB$ is a reduct of a finitely bounded homogeneous structure, so Conjecture~\ref{conj:tractability}
    implies that $\CSP(\bA^\bB)$ is in $\cP$. This implies that $\RCSP(\bA,\bB)$ is polynomial-time
    solvable (see, e.g., the last sentence of Theorem~\ref{thm:rpp-pp-constructions-finite-restriction}).
\end{proof}

\subsection{A small remark regarding PCSPs}

Theorem~\ref{thm:rpp-pp-constructions-finite-restriction} asserts in particular that
$\RCSP(\bA,\bB)$ is log-space equivalent to $\CSP(\bA^\bB)$. The following statement
strengthens this connection by showing that these two problems are further log-space
equivalent to $\PCSP(\bA,\bA^\bB)$.

\begin{lemma}\label{lem:RCSP-PCSP}
    For every possibly infinite structure $\bA$ and a (possibly infinite) structure and
    $\bB$ a finite structure, the following problems are polynomial-time equivalent.
    \begin{itemize}
        \item the constraint satisfaction problem $\CSP(\bA^\bB)$,
        \item the promise constraint satisfaction problem $\PCSP(\bA,\bA^\bB)$, and
        \item the restricted constraint satisfaction problem $\RCSP(\bA,\bB)$.
    \end{itemize}
\end{lemma}
\begin{proof}
    The first and third itemized problems are polynomial-time equivalent by
    Theorem~\ref{thm:rpp-pp-constructions-finite-restriction}. The promise CSP form the second
    item clearly reduces to the CSP in the first item. Now we see that the RCSP reduces
    to the PCSP. Let $\bC$ be an input
    to $\RCSP(\bA,\bB)$, so $\bC\to \bB$. We claim that $\bC$ is a valid input
    to $\PCSP(\bA,\bA^\bB)$: if $\bC\to \bA$, then $\bC$ is clearly a valid input to the
    PCSP; otherwise, notice that $\bC\times \bB$ is homomorphically equivalent to $\bC$
    because $\bC\to \bB$, and so if $\bC\not\to\bA$, then
    $\bC \times \bB \leftrightarrow \bC \not\to \bA^\bB$, and so by Lemma~\ref{lem:properties-exponential},
    we see that $\bC\not\to \bA^\bB$.
    Therefore, $\bC$ is a valid input to  $\PCSP(\bA,\bB)$, and clearly the reduction is complete
    and correct.
\end{proof}

\begin{corollary}
    For every non-bipartite graph $\bH$ and every finite digraph $\bD$ one of the following holds.
    \begin{itemize}
        \item Either $\bD\to \bH$, and in this case $\PCSP(\bH,\bH^\bD)$ is polynomial-time solvable, 
        \item or $\bD\not\to \bH$, and in this case $\PCSP(\bH,\bH^\bD)$ is $\NP$-hard.
    \end{itemize}
\end{corollary}
\begin{proof}
    The first item holds because if $\bD\to \bH$, then $\bH^\bD$ has a loop, and
    hence  $\PCSP(\bH,\bH^\bD)$ is trivial. Now suppose that $\bD\not\to\bH$. In this case,
    Theorem 3 in~\cite{brewsterDAM156} implies that $\RCSP(\bH,\bD)$ is $\NP$-hard, 
    and we conclude that $\PCSP(\bH,\bH^\bD)$ is $\NP$-hard (Lemma~\ref{lem:RCSP-PCSP}).
\end{proof}

\section{Homomorphism-free restrictions}
\label{sect:FO-restrictions}

As mentioned above, the reader familiar with monotone monadic strict NP (MMSNP) can notice
that digraph (finite domain) CSPs with input restriction to $\calF$-homomorphism-free
digraphs (structures) can be solved by infinite domain CSPs expressible in MMSNP.
Thus, these problems exhibit a P versus NP-complete dichotomy (see,
e.g.,~\cite{MMSNP-Journal,FederVardi}). Feder and Vardi reduce problems in MMSNP to
finite domain CSPs. Their reduction changes the input signature, and then uses the Sparse 
Incomparability Lemma. Here, we prove the following lemma which allows us
to preserve the signature by using duality pairs, and then, as Feder and Vardi, we proceed
via Sparse Incomparability.

\begin{lemma}\label{lem:xDT-T-homfree}
    For every finite set of structures $\calF$ and every finite digraph (structure) $\bB$,
    there is a finite set of structures $\calC$ such that the following problems are
    polynomial-time equivalent:
    \begin{itemize}
        \item deciding if an input structure $\bA$ belongs to $\CSP(\bC)$ for some $\bC\in \calC$, and
        \item $\CSP(\bB)$ restricted to $\calF$-homomorphism-free structures.
    \end{itemize}
\end{lemma}
\begin{proof}
    We first consider the case when for every $\bF\in\calF$ there is no forest $\bT$
    such that $\bF\to \bT$. In this case, let $\calC = \{\bB\}$. Clearly, $\CSP(\bB)$ is at least as
    hard as $\CSP(\bB)$ restricted to  $\calF$-homomorphism-free inputs. 
    For the converse reduction, let $\ell$ be the maximum number of vertices of a
    structure $\bF\in\calF$, and $k$ the number of vertices of $\bB$. On input $\bA$ to $\CSP(\bB)$, 
    let $\bA'$ be the structure of girth larger than $\ell$ obtained via Theorem~\ref{thm:sparse-incomparability}.
    In particular, $\bA'$ can be constructed in polynomial-time from $\bA$ and $\bA'\to \bB$ if and only
    if $\bA\to \bB$. Finally, since the girth of $\bA'$ is strictly larger than the number of vertices
    of any structure $\bF\in \calF$, and no $\bF\in\calF$ maps to a forest, then $\bA'\in\Forb(\calF)$. 
    Hence, $\bA'$ is a valid input to the second itemized problem, and thus 
    $\CSP(\bB)$ is polynomial-time equivalent to $\CSP(\bB)$ restricted to $\calF$-homomorphism-free structures.
    
    Otherwise, let $\calT$ be the non-empty set of  forests $\bT$ for which there is an $\bF\in\calF$
    and a surjective homomorphism $f\colon \bF\to \bT$. By Theorem~\ref{thm:gen-dualities}, there is a
    finite set of structures $\calD$ such that $(\calT, \calD)$ is a generalized duality pair.
    Let $\calC$ be the set $\{\bB\times \bD\colon \bD\in \calD\}$. We now prove that the claim of this lemma
    holds for $\calC$ and $\bB$.
    
    Consider an input structure $\bA$ to the second itemized problem. Since $\bA\in \Forb(\calF)$, 
    it must also be the case that $\bA \in \Forb(\calT)$, and so there is some $\bD\in \calD$ such that
    $\bA\to \bD$. Hence, $\bA\to \bB$ if and only if $\bA\to \bD\times \bB$, and so $\bA\to \bB$ implies
    that $\bA \in \CSP(\bC)$ for some $\bC\in \calC$. Conversely, if $\bA \in \CSP(\bC)$
    for some $\bC\in \calC$, then $\bA\to \bB\times \bD$ for some $\bD\in\calD$, and therefore $\bA\to \bB$. 
    We thus conclude that the problem in the second item reduces in polynomial time
    (via the trivial reduction) to deciding if an input structure $\bA$ belongs to $\bigcup_{\bC\in \calC}\CSP(\bC)$.
    
    For the reduction back, we consider two possible cases. First, suppose that
    $\bD\to \bB$ for every $\bD\in \calD$, and notice that in this case
    $\bB\times \bD$ is homomorphically equivalent to $\bD$ for every $\bD\in \calD$.
    Hence, $\CSP(\bC)$ is polynomial-time solvable for every for every $\bC\in\calC$, 
    and so the problem in the first item is polynomial-time solvable as well. Notice that
    if we show that in this case, $\CSP(\bB)$ restricted to $\calF$-homomorphism-free structures
    is polynomial-time solvable, the both itemized problems are clearly
    polynomial-time equivalent. Similarly as above, if an input structure $\bA$ to $\CSP(\bB)$
    has no homomorphic image from any $\bF\in\calF$, then $\bA\in\Forb(\calT)$ and so $\bA\to \bD_0$
    for some $\bD_0\in \calD$.  Since $\bD\to \bB$ for every $\bD\in \calD$, we see that
    $\bA\to \bD_0\to \bB$, 
    i.e., every input $\bA$ to  $\CSP(\bB)$ with no homomorphic image of any $\bF\in\calF$
    is a yes-instance to $\CSP(\bB)$. Therefore, if $\bD\to \bB$ for every $\bD \in\calD$,
    both itemized problems are polynomial-time solvable and thus, polynomial-time equivalent. 

    Now, suppose that $\bD_0\not \to \bB$ for some $\bD_0\in \calD$, and recall that there is no
    forest $\bT\in\calT$ such that $\bT\to \bD_0$.  With a similar trick as before, we can
    use the Sparse Incomparability Lemma to find a structure $\bD_0'$ such that $\bD_0'\not\to \bB$
    and $\bD_0$ has no homomorphic image from any $\bF\in\calF$.  The reduction from the first
    itemized problem for $\calC$ to $\CSP(\bB)$ to structures with no homomorphic image of any
    $\bF\in\calF$ is now simple:
    on input $\bA$ we check if $\bA\in \Forb(\calT)$, if not, we return $\bD_0'$, 
    and if yes we return $\bA'$ obtained from the Sparse Incomparability Lemma applied to $\bA$, 
    to $\ell$ the maximum number of vertices of a digraph in $\calF$, and $k:=|V(\bB)|$.
    The fact that this reduction is consistent and correct follows by the assumption that
    $\bD_0'\not\to \bB$,  and with similar arguments as in the first paragraph, one can notice
    that both $\bA'$ and $\bD_0$ are valid inputs to the problem in the second item.
\end{proof}

The following statement is an immediate consequence of Lemma~\ref{lem:xDT-T-homfree}
and finite domain CSP dichotomy~\cite{BulatovFVConjecture,Zhuk20}.

\begin{theorem}\label{thm:FO-restructions-complexity-dichotomy}
    For every finite digraph (structure) $\bH$ and every finite set of digraphs (structures)
    $\calF$, $\CSP(\bH)$ restricted to $\calF$-homomorphism-free digraphs (structures) 
    is either in $\cP$ or it is  $\NP$-complete.
\end{theorem}
\begin{proof}
    Let $\calC$ be the finite set of structures from Lemma~\ref{lem:xDT-T-homfree}
    applied to $\CSP(\bA)$ and to $\calF$.
    If $\CSP(\bC)$ is polynomial-time solvable for every $\bC\in\calC$,
    then deciding if an input structure belongs to $\bigcup_{\bC\in \calC}\CSP(\bC)$
    is polynomial-time solvable, and so $\CSP(\bB)$ is polynomial-time solvable
    for $\calF$-homomorphism-free structures.
    Otherwise, it follows from the dichotomy for finite domain CSPs~\cite{BulatovFVConjecture,Zhuk20}
    than $\CSP(\bC_0)$ is $\NP$-complete for some $\bC_0\in \calC$. It straightforward to observe
    that in this case deciding if an input structure belongs to $\bigcup_{\bC\in \calC}\CSP(\bC)$ is
    $\NP$-complete as well (see, e.g.,~\cite[Theorem 25]{FoniokNesetril}), 
    and hence $\CSP(\bB)$ is $\NP$-hard even for $\calF$-homomorphism-free structures.
    The claim now follows because $\Forb(\calF) = \CSP(\bB)$.
\end{proof}

\section{Acyclic digraphs and bounded paths}
\label{sect:acyclic}

In this section we apply our results above, and the theory of constraint satisfaction to
obtain results in the context of $\calF$-(subgraph)-free algorithmics. 
The main result of this section settles Question~\ref{qst:acyclic}
(Theorem~\ref{thm:acyclic+bounded-paths}). Moreover, we show that 
if $\CSP(\bH)$ is $\NP$-complete, then there is a positive integer $N$ such that
$\CSP(\bH)$ remains $\NP$-complete for acyclic digraphs with no directed path on
$N$ vertices, and $\CSP(\bH)$ can be solved in polynomial time if the input is an
acyclic digraph with no directed path on $N-1$ vertices.
We begin by stating the following corollary of (the proof of) Lemma~\ref{lem:xDT-T-homfree}.

\begin{corollary}\label{cor:trees-hom-free}
    For every finite set of trees $\calF$ with dual $\bD$, and every digraph $\bH$ the following 
    problems are polynomial-time equivalent.
    \begin{itemize}
        \item $\RCSP(\bH,\bD)$.
        \item $\CSP(\bH\times \bD)$.
        \item $\CSP(\bH)$ restricted to $\calF$-homomorphism-free digraphs.
    \end{itemize}
\end{corollary}

A \emph{polymorphism} of a structure $\bA$ is a homomorphism $f\colon \bA^n\to \bA$, and in 
this case we say that $f$ is an \emph{$n$-ary} polymorphism. Under composition, polymorphisms
form an algebraic structure called a \emph{clone}, and this structure captures the
computational complexity of $\CSP(\bA)$ (see, e.g.,\cite{Pol}). A $4$-ary polymorphism
$f\colon \bA^4\to \bA$ satisfies the \emph{Sigger's identity} if 
\[
    f(x_1,x_2,x_3,x_1) = f(x_2,x_1,x_2,x_3) \text{ for every } x_1,x_2,x_3,x_4\in A.
\]
The Sigger's identity is one of several identities that equivalently describe the tractability
frontier for finite domain CSPs.

\begin{theorem}
[Equivalent to Theorem 1.4 in~\cite{Zhuk20}]
\label{thm:finite-domain-dichotomy}
    For every finite structure $\bA$ one of the following holds.
    \begin{itemize}
        \item Either $\bA$ has a polymorphism $f\colon \bA^4\to \bA$ that satisfies the
        Sigger's identity, and in this case $\CSP(\bA)$ is polynomial-time solvable, or
        \item otherwise, $\bA$ pp-constructs $\bK_3$, and in this case $\CSP(\bA)$ is $\NP$-complete.
    \end{itemize}
\end{theorem}

Consider a finite digraph $\bH$, and let $k$ be a positive integer. For a given function 
$f\colon (H\times [k])^n\to (H\times [k])$ and for each $i\in [k]$ we define a function
\[
f_i\colon H^n\to H \text{ by } (h_1,\dots, h_n)\mapsto \pi_H(f((h_1,i),\dots, (h_n,i)),
\]
where $\pi_H$ is the projection of $H\times [k]$ onto $H$.

\begin{lemma}\label{lem:HxTk-loop-condition}
    Let $\bH$ be a digraph, $k$ a positive integer, and $f\colon (\bH\times \bT_k)^n\to (\bH\times \bT_k)$
    and $n$-ary polymorphism of $\bH\times \bT_k$. If there are $i,j\in [k]$ with $i < j$ such that
    $f_i = f_j$, then $f_i$ is an $n$-ary polymorphism of $\bH$, and if $f$ satisfies the equalities
    \[
        f(x_{\sigma(1)},\dots, x_{\sigma(n)}) = f(x_{\rho(1)},\dots, x_{\rho(n)}) \text{ for all } x_1,\dots, x_m\in H\times [k]
    \]
    for some $\sigma,\rho\colon[n]\to [m]$, then $f_i$ satisfies the equalities  
    \[
    f_i(h_{\sigma(1)},\dots, h_{\sigma(n)}) = f_i(h_{\rho(1)},\dots, h_{\rho(n)}) \text{ for all } h_1,\dots, h_n\in H.
    \]
\end{lemma}
\begin{proof}
    We first see that $f_i$ is an $n$-ary polymorphism of $\bH$. Let $((h_1,\dots, h_n),(h'_1,\dots, h'_n))$
    be an edge of $\bH^n$, i.e., $(h_m,h'_m)$ is an edge for $\bH$ for each $m\in [n]$. Since $i < j$, it follows
    that $((h_1,i),\dots, (h_n,i)),((h'_1,j),\dots, (h'_n,j))$ is an edge of $\bH\times \bT_n$. Since $f$ is a
    polymorphism, it must be the case that $(f((h_1,i),\dots, (h_n,i)),f((h'_1,j),\dots, (h'_n,j)))$ is an edge
    of $\bH\times \bT_n$. Recall that the projection $\pi_H$ is a homomorphism from $\bH\times \bT_n$ onto $\bH$,
    and so
    \[
    (f_i(h_1,\dots, h_n),f_j(h'_1,\dots, h'_n)) = (\pi_Hf((h_1,i),\dots, (h_n,i)),\pi_Hf((h'_1,j),\dots, (h'_n,j))) \in E(\bH).
    \]
    Finally, since $f_i = f_j$, it follows that $(f_i(h_1,\dots, h_n),f_j(h_1,\dots, h_n))$ is an edge
    of $\bH$, and thus $f_i\colon \bH^n\to \bH$ is an $n$-ary polymorphism of $\bH$.

    Now, let $h_1,\dots, h_m\in H$, and for each $l\in[n]$ let $\bar{h}_l := (h_l,i)$, so
    \[
    f_i(h_{\sigma(1)},\dots, h_{\sigma(n)}) = \pi_H(f(\bar{h}_{\sigma(1)},\dots, \bar{h}_{\sigma(n)}).
    \]
    Since $f$ satisfies the loop condition for $\sigma$ and $\rho$, it follows that
    \[
    \pi_H(f(\bar{h}_{\sigma(1)},\dots, \bar{h}_{\sigma(n)}) = \pi_H(f(\bar{h}_{\rho(1)},\dots, \bar{h}_{\rho(n)})
    = f_i(h_{\rho(1)},\dots, h_{\rho(n)}).
    \]
    And thus, $f_i(h_{\sigma(1)},\dots, h_{\sigma(n)}) = f_i(h_{\rho(1)},\dots, h_{\rho(n)})$, and since the
    choice  of $h_1,\dots, h_m\in H$ was arbitrary, the claim of this lemma follows.
\end{proof}

Building on this lemma and Theorem~\ref{thm:finite-domain-dichotomy} together with
Corollary~\ref{cor:trees-hom-free}, we prove the main result of this section.

\begin{theorem}\label{thm:RCSP(H,TTk)}
    For every finite digraph $\bH$ one of the following statements holds.
    \begin{itemize}
        \item Either $\bH$ has a polymorphism satisfying the Sigger's identity, and $\RCSP(\bH,\bT_k)$ is in $\cP$ for
        every positive integer $k$, or
        \item otherwise, there is a positive integer $N$ such that  $\RCSP(\bH,\bT_k)$ is $\NP$-hard for every $k\ge N$.
    \end{itemize}
\end{theorem}
\begin{proof}
    The first item holds by Theorem~\ref{thm:finite-domain-dichotomy}, and because
    if $\CSP(\bH)$ is in $\cP$, then $\RCSP(\bH,\bT_k)$ is in $\cP$ for ever positive integer
    $k$. Now, suppose that $\bH$ does not have a polymorphism satisfying the Sigger's
    identity. It follows from Corollary~\ref{cor:trees-hom-free} that $\RCSP(\bH,\bT_k)$ and
    $\CSP(\bH\times \bT_k)$ are polynomial-time equivalent. Let $M = |H|^{4|H|}+1$,
    and consider a polymorphism $f\colon (\bH\times \bT_N)^\to \bH$. Clearly, there
    are $M-1$ functions from $H^4$ to $H$, and since for each $i\in[M]$ every $f_i$
    is a  $4$-ary function of $H$, it follows from the choice of $M$ that there are
    $i < j \le M$  such that $f_i = f_j$. So, by Lemma~\ref{lem:HxTk-loop-condition}
    if $f$ satisfies the Sigger's identity, then $f_i$ satisfies the Sigger's identity.
    Since $\bH$ does not have such a polymorphism, then $\bH\times \bT_M$ does not have
    such a polymorphism, and hence $\CSP(\bH\times \bT_k)$ is $\NP$-hard (see, e.g., 
    Theorem~\ref{thm:finite-domain-dichotomy}). Moreover, since $\CSP(\bH\times \bT_M)$
    is polynomial-time equivalent to $\RCSP(\bH,\bT_M)$, it follows that 
    $\RCSP(\bH,\bT_M)$. Finally, it is straightforward to observe that
    if $\RCSP(\bH, \bT_n)$ is at least as hard as $\RCSP(\bH,\bT_{n+1})$,
    and so, if $N$ is the smallest integer such that $\bH\times \bT_N$ does
    not have a Sigger's polymorphism, then $N$ witnesses that the second
    itemized statement holds.
\end{proof}

The following is an immediate application of this theorem, of its proof, and of
Corollary~\ref{cor:trees-hom-free}.

\begin{corollary}\label{cor:CSPH-forbidden-walks}
    For every finite digraph $\bH$, there is a positive integer $N\le |H|^{4|H|}+1$ such that
    $\CSP(\bH)$ is polynomial-time equivalent to $\CSP(\bH)$ restricted to acyclic digraphs with no
    directed walk on $N$ vertices. 
\end{corollary}

As promised, we apply the framework of RCSPs to the context of $\calF$-(subgraph)-free
algorithmics.

\begin{theorem}\label{thm:acyclic+bounded-paths}
 For every digraph $\bH$ such that $\CSP(\bH)$ is $\NP$-hard, there is a positive integer
    $N$ such that $\CSP(\bH)$ remains $\NP$-complete even for $\vec{\bP}_N$-subgraph-free
    acyclic digraphs. Moreover, there is such an $N$ such that $\CSP(\bH)$ is polynomial-time
    solvable when the input is a $\vec{\bP}_{N-1}$-subgraph-free digraph.
\end{theorem}
\begin{proof}
    If $\cP = \NP$ the claim is trivial, so we assume that $\cP\neq \NP$. Notice that
    a digraph $\bD$ is acyclic and has no directed path on $N$ verties if and only if
    there is no homomorphism $\vec{\bP}_N\to \bD$. Also, since $\CSP(\bH)$ is $\NP$-complete
    and we are assuming that $\cP\neq \NP$, it follows from Theorem~\ref{thm:finite-domain-dichotomy}
    that $\bH$ has no Sigger's polymorphism. The claim of this theorem now follows from these
    simple arguments and Theorem~\ref{thm:RCSP(H,TTk)}.
\end{proof}

\begin{corollary}\label{thm:digraph-}
    Let $\bH$ be a digraph such that $\CSP(\bH)$ is $\NP$-hard. Then, there are positive
    integers $N$ and $M$ such that
    \begin{itemize}
        \item $\CSP(\bH)$ is $\NP$-hard for $\vec{\bP}_k$-free digraphs whenever $k\ge N$, and
        \item $\CSP(\bH)$ is $\NP$-hard for $\vec{\bP}_k$-subgraph-free digraphs whenever $k\ge M$.
    \end{itemize}
\end{corollary}

We conclude this section with a simple example showing that Theorem~\ref{thm:acyclic+bounded-paths}
does not hold when $\bH$ is infinite: there are infinite graphs $\bH$ such that 
$\CSP(\bH)$ is $\NP$-complete, but $\CSP(\bH)$ becomes tractable on acyclic instances.
Moreover, $\bH$ can be chosen to be \emph{$\omega$-categorical}, i.e., 
for every positive integer $k$ the automorphism group of $\bH$ defines finitely many orbits
of $k$-tuples.

\begin{example}\label{ex:infinite-acyclic}
    Let $\bH$ be the disjoint union of $\bK_3$ and the rational number with the strict linear order. 
    It is straightforward to observe that \COL{3} reduces to $\CSP(\bH)$, and that
    every acyclic digraph $\bD$ is a yes-instance of $\CSP(\bH)$. Hence, 
    $\CSP(\bH)$ is $\NP$-complete, but it is polynomial-time solvable on acyclic instance.
    The fact that $\bH$ is $\omega$-categorical follows from its definition --- it is the
    disjoin union of two $\omega$-categorical digraphs.
\end{example}

\section{Digraphs on three vertices}
\label{sect:3-vertices}

In this section we answer Question~\ref{qst:Pk} for digraphs on three vertices.
A loopless digraph on three vertices either contains two directed cycles and its CSP if $\NP$-complete,
or otherwise its CSP is polynomial-time solvable. Hence, we focus on the three loopless digraphs
on three vertices: $\vec{\bC}_3^+$ (obtained from the directed cycle by adding one edge),
$\vec{\bC}_3^{++}$ (obtained from the directed cycle by adding two edges), and $\bK_3$ --- see
also Figure~\ref{fig:three-vertices}.

Also note that these three digraphs are hereditarily hard (see Theorem~\ref{thm:smooth-hereditay}), 
and it thus follows from Theorem~\ref{thm:brewster} that $\RCSP(\bH,\bT_4)$ is $\NP$-hard
for $\bH\in\{\vec{\bC}_3^+, \vec{\bC}_3^{++}, \bK_3\}$.
In turn, this implies that for such digraphs $\bH$ the problem $\CSP(\bH)$ is $\NP$-hard even
restricted to $\vec{\bP}_5$-(subgraph)-free digraphs. We use this remark in both subsections below.

\begin{figure}[ht!]
\centering
\begin{tikzpicture}

  \begin{scope}[xshift = -4.5cm, scale = 0.8]
   \node [vertex, label = above:$1$] (1) at (90:1.2) {};
    \node [vertex, label = left:$3$] (2) at (210:1.2) {};
    \node [vertex, label = right:$2$] (3) at (330:1.2) {};
    \node (L1) at (0,-1.5) {$\vec{\bC}_3^+$};
      
    \foreach \from/\to in {1/2, 2/3, 3/1,  3/2}     
    \draw [arc] (\from) to [bend right = 20] (\to);
  \end{scope}

  \begin{scope}[scale = 0.8]
    \node [vertex, label = above:$1$] (1) at (90:1.2) {};
    \node [vertex, label = left:$3$] (2) at (210:1.2) {};
    \node [vertex, label = right:$2$] (3) at (330:1.2) {};
      
    \foreach \from/\to in {1/2, 2/3, 3/1, 1/3, 3/2}     
    \draw [arc] (\from) to [bend right = 20] (\to);
  \end{scope}

  \begin{scope}[xshift = 4.5cm, scale = 0.8]
    \node [vertex, label = above:$1$] (1) at (90:1.2) {};
    \node [vertex, label = left:$3$] (2) at (210:1.2) {};
    \node [vertex, label = right:$2$] (3) at (330:1.2) {};
    \node (L1) at (0,-1.5) {$\bK_3$};
      
    \foreach \from/\to in {1/2, 2/3, 3/1, 1/3, 3/2, 2/1}       
    \draw [arc] (\from) to [bend right = 20] (\to);
  \end{scope}

\end{tikzpicture}
\caption{The three digraphs on three vertices with at least two directed cycles, and whose CSP is NP-complete.}
\label{fig:three-vertices}
\end{figure}
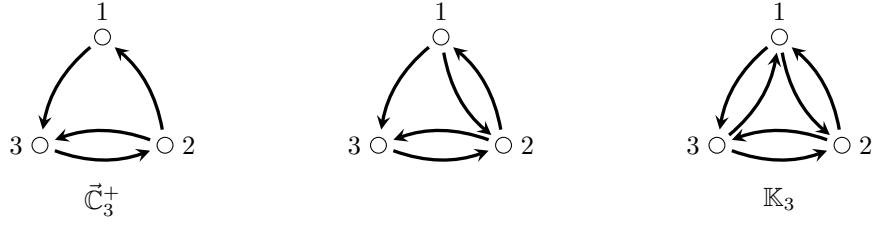

\subsection{$\vec{\bP}_k$-subgraph-free digraphs}

It is straightforward to observe that any orientation of an odd cycle contains
$\vec{\bP}_3$ as a subgraph. Hence, if $\bD$ is a $\vec{\bP}_3$-subgraph-free digraph, 
then it is bipartite, i.e., $\bD\to \bK_2$. In particular, this implies that 
$\CSP(\vec{\bC}_3^+)$, $\CSP(\vec{\bC}_3^{++})$, and $\CSP(K_3)$ are polynomial-time solvable for
$\vec{\bP}_3$-subgraph-free digraphs. As previously mentioned, each of these CSPs is NP-hard
for $\vec{\bP}_5$-subgraph-free digraphs. In this section, we study the complexity
$\CSP(\vec{\bC}_3^+)$, $\CSP(\vec{\bC}_3^{++})$, and $\CSP(\bK_3)$ restricted to
$\vec{\bP}_4$-subgraph-free digraphs. We begin with a remark which we use
a couple times in the remaining of the paper.

\begin{remark}\label{rmk:3-cycle}
    If a $\vec{\bP}_4$-subgraph-free weakly connected digraph $\bD$ contains
    (as a subgraph) a directed $3$-cycle $v_1,v_2,v_3$, then $D = \{v_1,v_2,v_3\}$.
    Indeed, since $\bD$ is $\vec{\bP}_4$-subgraph-free and $(v_1,v_2),(v_2,v_3),(v_3,v_1)\in E$,
    it must be the case that the out- and in-neighbourhood of each $v_i$ is a subset of
    $\{v_1,v_2,v_3\}$ (and so, the claim follows because $\bD$ is weakly connected).
\end{remark}

\begin{observation}\label{obs:K3-P4-subgraph-free}
    Every loopless $\vec{\bP}_4$-subgraph-free digraph is $3$-colourable.
\end{observation}
\begin{proof}
    We claim that it suffices to consider oriented graphs.
    Indeed, if $\bD$ is not an oriented digraph, let $\bD'$ be any spanning subdigraph
    of $\bD$ obtained by removing exactly one edge of each symmetric pair $(x,y),(y,x)\in E(\bD)$. 
    Clearly, $\bD$ is $3$-colourable if and only if $\bD'$ is $3$-colourable,
    and moreover, if $\bD$ is a loopless $\vec{\bP}_4$-subgraph-free digraph, then $\bD'$
    is loopless $\vec{\bP}_4$-subgraph-free oriented graph. So assume that $\bD$ is
    a loopless $\vec{\bP}_4$-subgraph free oriented digraph, and without loss of generality
    consider the case when $\bD$ is weakly connected. If $\bD$ contains a directed $3$-cycle
    $x_1,x_2,x_3$, then $D = \{x_1,x_2,x_3\}$ (Remark~\ref{rmk:3-cycle}), and hence $\bD\to \bK_3$. 
    Otherwise, assume that $\bD$ contains no $3$-cycle, and notice that in this case
    it follows that  $\vec{\bP}_4\not\to \bD$. Hence, by Observation~\ref{obs:P-TT}
    we conclude that $\bD\to \bT_3\to \bK_3$.
\end{proof}

Clearly, Observation~\ref{obs:K3-P4-subgraph-free} does not hold if we replace $\CSP(\bK_3)$ (i.e.,
$3$-colourable) by $\CSP(\vec{\bC}_3^{++})$: $\bK_3$ is a simple counterexample. However, it turns
out that $\bK_3$ is the unique minimal counterexample. In the following proof we
use the notion of a \emph{leaf} of a digraph $\bD$, i.e., a vertex $x\in \bD$ with
$|(N^+(x)\cup N^-(x))\setminus\{x\}| = 1$.

\begin{lemma}\label{lem:C3++-P4-subgraph-free}
    Every loopless $\{\bK_3,\vec{\bP}_4\}$-subgraph-free digraph $\bD$ admits a homomorphism
    to $\vec{\bC}_3^{++}$.
\end{lemma}
\begin{proof}
    Without loss of generality assume that $\bD$ is a weakly connected digraph. 
    If $\bD$ contains $\vec{\bC}_3$ as a subgraph, then the claim follows from Remark~\ref{rmk:3-cycle}.
    Also, notice that if $\bD$ contains a symmetric path on three vertices $x_1,x_2,x_3$, i.e.,
    $x_1\neq x_3$ and $(x_1,x_2),(x_2,x_1),(x_2,x_3),(x_3,x_2)$, then
    $x_2$ is the unique neighbour of $x_1$ and also the unique neighbour of $x_3$ (because
    $\bD$ is $\vec{\bP}_4$-free). In other words, $x_1$ and $x_3$ are leaves of $\bD$
    with the same neighbourhood, and so, the digraph $\bD$ homomorphically maps to
    $\bD-x_3$, and thus $\bD\to \vec{\bC}_3^{++}$ if and only if $(\bD-x_3)\to \vec{\bC}_3^{++}$.
    Hence, it suffices to prove the claim for weakly connected $\{\vec{\bC_3},\bP_3,\vec{\bP}_4\}$-subgraph-free
    digraphs, and since every vertex of $\vec{\bC}_3^{++}$ is incident to a symmetric
    edge, we may further assume that $\bD$ has no leaves.

    Let $\bD^\ast$ be any digraph obtained from $\bD$ after removing exactly one edge from
    every symmetric pair of edges of $\bD$. Since $\bD$ is $\{\vec{\bC}_3,\vec{\bP}_4\}$-subgraph-free,
    $\bD^\ast$ has no directed walk on four vertices, and so there is a homomorphism
    $f\colon \bD^\ast\to \bT\bT_3$ (by Observation~\ref{obs:P-TT}). We claim that  the 
    same function $f$ defines a homomorphism $f\colon \bD\to \vec{\bC}_3^{++}$. To do so, 
    it suffices to show that there is no edge $(x,y)$ such that $f(x) = 3$ and $f(y) = 1$.
    Notice that if this were the case, since $f\colon \bD^\ast\to \bT\bT_3$ is a homomorphism,
    then $x$ and $y$ would induce a symmetric pair of edges in $\bD$, i.e., $(x,y),(y,x)\in E(\bD)$.
    Thus, it suffices to show that if $(x,y),(y,x)$ is a symmetric pair of edges, then 
    $f(x) =2$ or $f(y) = 2$. Recall that we assume that $\bD$ has no leaves, hence
    $x$ and $y$ have some neighbours $x'$ and $y'$ respectively. Without loss of
    generality assume that $(x,y)\in E(\bD^\ast)$, and since $\bD$ is
    $\vec{\bP}_4$-subgraph-free, either $(x,x'),(y,y')\in E(\bD)$ or
    $(x',x),(y',y)\in E(\bD)$. Notice that in the former case, $y$ is the middle vertex
    of a directed path of length $2$ in $\bD^\ast$. Hence, any homomorphism from $\bD^\ast$
    to $\bT\bT_3$ maps $y$ to $2$, so in particular, $f(y) = 2$. With symmetric
    arguments it follows that if $(x',x),(y',y)\in E(\bD)$, then $f(x) = 2$.
    This proves that the function $f\colon D\to \{1,2,3\}$ defines a homomorphism
    $f\colon \bD\to \vec{\bC}_3^\ast$.
\end{proof}

Now we show that $\CSP(\vec{\bC}_3^+)$ is polynomial-time solvable when the input is restricted
to $\vec{\bP}_4$-subgraph-free digraphs. To do so, we consider a unary predicate $U$, and we
reduce the problem mentioned above to the CSP of the $\{E,U\}$-structure depicted  in 
Figure~\ref{fig:digraph+unaries}. We first show that the underlying graph has a \emph{conservative
majority} polymorphism, i.e., a polymorphism $f\colon \bG^3\to \bG$ such that $f(x,x,x) = f(x,x,y) = f(x,y,x) = f(y,x,x) = x$
and $f(x,y,z)\in \{x,y,z\}$.

\begin{lemma}\label{lem:conservative-majority}
    The digraph $\bG$ from Figure~\ref{fig:digraph+unaries} has a conservative majority polymorphism. 
    In particular, $\CSP(\bG,U)$ can be solved in polynomial time.
\end{lemma}
\begin{proof}  
    We define a conservative majority function $f\colon G^3\to G$, and we then argue that it is
    a polymorphism. We denote by $\pi_1$ the projection onto the first coordinate, and by
    $\maj(x,y,z)$ the majority operation when $|\{x,y,z\}|\le 2$. Also, we simplify our writing
    by implicitly assuming that in the $n$-th itemized case of the following definition, neither
    of the first $n-1$ cases holds.
    \begin{equation*}
        f(x,y,z)=
            \begin{cases}
                \maj(x,y,z) & \text{if } |\{x,y,z\}| \le 2,\\
                1 & \text{if }   1\in \{x,y,z\}, \text{ and } \{3,7\}\cap \{x,y,z\}\neq\varnothing,\\
                2 & \text{if }   2\in \{x,y,z\}, \text{ and } \{6,7\}\cap \{x,y,z\}\neq\varnothing,\\
                3 & \text{if }   3\in \{x,y,z\},  \text{ and } \{5,7\}\cap\{x,y,z\}\neq\varnothing,\\
                4 & \text{if }   4,6\in \{x,y,z\},\\
                6 & \text{if } \{x,y,z\} = \{5,6,7\},\\
                \pi_1(x,y,z) & \text{otherwise.}\\
    \end{cases}
\end{equation*}
    It is clear that $f$ is a symmetric conservative majority function. It is routine  verifying that
    $f$ is indeed a polymorphism.
    The fact that $\CSP(\bG,U)$ can be solved in polynomial time now follows from~\cite{DalmauKrokhin08} because the same function $f$ defines a majority polymorphism of $(\bG,U)$. 
\end{proof}

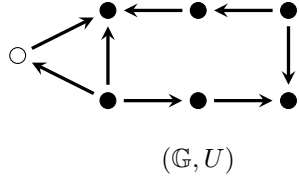
\begin{figure}[ht!]
\centering
\begin{tikzpicture}

  \begin{scope}[xshift = -4.5cm, scale = 0.8]
    \node [vertex, fill = black] (1) at (-1.5,1.5) {};
    \node [vertex, fill = black] (2) at (0,1.5) {};
    \node [vertex, fill = black] (3) at (1.5,1.5) {};
    \node [vertex, fill = black] (4) at (1.5,0) {};
     \node [vertex, fill = black] (5) at (0,0) {};
      \node [vertex, fill = black] (6) at (-1.5,0) {};
       \node [vertex] (7) at (-3,0.75) {};
    \node (L1) at (0,-1) {$(\bG,U)$};
      
    \foreach \from/\to in {3/2, 2/1, 3/4, 5/4, 6/5, 6/1, 6/7, 7/1}     
    \draw [arc] (\from) to (\to);
  \end{scope}

\end{tikzpicture}
\caption{A $\{E,U\}$-structure where $E$ is a binary relation represented by arcs, and $U$ is a unary relation represented by black vertices.}
\label{fig:digraph+unaries}
\end{figure}

\begin{lemma}\label{lem:C3+-P4-subgraph-free}
    $\CSP(\vec{\bC}_3^+)$ can be solved in polynomial-time when the input is restricted
    to $\vec{\bP}_4$-subgraph-free digraphs.
\end{lemma}
\begin{proof}
    Consider a loopless weakly connected digraph $\bD$ with no directed $3$-cycles.
    Let $B$ be the subset of vertices of $\bD$ incident to some symmetric edge, and let
    $\bD^\ast$ be the digraph obtained from $\bD$ after removing exactly one edge from
    each symmetric pair of edges. It is straightforward to observe there is no homomorphism
    from the directed path on four vertices to $\bD^\ast$. Hence, by Observation~\ref{obs:P-TT},
    $(\bD^\ast,B)$ homomorphically maps to $(\bT\bT_3,\{1,2,3\})$ considered as a  $\{E,U\}$-structure
    where $E$ is a binary predicate and $U$ a unary predicate (this simply means that there is a
    homomorphism $\bD^\ast\to \bT\bT_3$ such that every vertex in $B$ is mapped to a vertex in $\{1,2,3\}$).
    Also note that $\bD\to \vec{\bC}_3^+$ if and only if there is a homomorphism $f\colon \bD^\ast\to \vec{\bC}_3^+$
    where $f(B)\subseteq \{2,3\}$, i.e., $\bD \to \vec{\bC}_3^+$ if and only if 
    $(\bD^\ast,B)\to (\vec{\bC}_3^+, \{2,3\})$ (as $\{E,U\}$-structures). It follows from these arguments
    that $\CSP(\vec{\bC}_3^+)$ restricted to $\{\vec{\bC}_3,\vec{\bP}_4\}$-subgraph-free digraphs
    reduces in polynomial time to $\CSP((\vec{\bC}_3^+, \{2,3\}) \times (\bT\bT_3,\{1,2,3\}))$.
    It is not hard to observe that the core of $(\vec{\bC}_3^+, \{2,3\}) \times (\bT\bT_3,\{1,2,3\})$
    is the structure $(\bG,U)$  depicted in Figure~\ref{fig:digraph+unaries}. It follows from
    Lemma~\ref{lem:conservative-majority} that  $\CSP(\bG,U)$ can be solved in polynomial time. 
 
    An algorithm that solves $\CSP(\vec{\bC}_3^+)$ in polynomial time when the input is restricted
    to $\vec{\bP}_4$-subgraph-free digraphs processes each weakly connected component of the input
    $\bD$ and accepts if and only if the following subroutine accepts each weakly connected component.
    On a weakly connected component $\bD'$ of $\bD$, the subroutine distinguishes whether the $\bD'$
    is $\vec{\bC}_3$-subgraph-free: if not,  it accepts if $\bD'$ is either $\vec{\bC}_3$ or $\vec{\bC}_3^+$,
    and rejects otherwise --- this step is consistent and correct because $|D'| = 3$ (see Remark~\ref{rmk:3-cycle}); if
    $\bD'$ is $\vec{\bC}_3$-free, then the subroutine accepts or rejects $\bD'$
    according to the reduction to $\CSP(\bG,U)$ explained above. Given the arguments in the
    previous paragraph, it is clear that this algorithm is consistent, correct, and runs in
    polynomial time. 
\end{proof}

\begin{theorem}\label{thm:3-vertices-Pk-subgraph-free}
    The following statements hold for every positive integer $k$.
    \begin{itemize}
        \item If $k\le 4$, then $\CSP(\vec{\bC}_3^+)$, $\CSP(\vec{\bC}_3^{++})$, and $\CSP(\bK_3)$ are solvable in polynomial-time
        when the input is restricted to $\vec{\bP}_k$-subgraph-free digraphs. 
        \item If $k\ge 5$, then $\CSP(\vec{\bC}_3^+)$, $\CSP(\vec{\bC}_3^{++})$, and $\CSP(\bK_3)$ are $\NP$-hard even
        if the input is restricted to $\vec{\bP}_k$-subgraph-free digraphs.
    \end{itemize}
\end{theorem}
\begin{proof}
    The first itemized statement follows from Observation~\ref{obs:K3-P4-subgraph-free}, from
    Lemma~\ref{lem:C3++-P4-subgraph-free} and from Lemma~\ref{lem:C3+-P4-subgraph-free}. The second
    itemized claim follows from the  discussion in the second paragraph of Section~\ref{sect:3-vertices}.
\end{proof}

\subsection{$\vec{\bP}_k$-free digraphs}

Notice that a digraph $\bD$ is $\vec{\bP}_2$-free if and only if it is a symmetric (undirected) graph.
Hence,  $\CSP(\bK_3)$ is $\NP$-hard for $\vec{\bP}_2$-free digraphs. Also note that a symmetric
graph $\bD$ maps to $\vec{\bC}_3^+$ if and only if $\bD$ is bipartite, and the same statement holds
for $\vec{\bC}_3^{++}$. Hence, in this section we study the complexity of $\CSP(\vec{\bC}_3^+)$
and $\CSP(\vec{\bC}_3^{++})$ with input restrictions to $\vec{\bP}_3$-free and to 
$\vec{\bP}_4$-free digraphs.

\begin{lemma}\label{lem:C3+-P4-free}
    $\CSP(\vec{\bC}_3^+)$ is $\NP$-hard even when the input
    $\bD$ satisfies the following conditions, where $\bF$ is any given orientation of the claw $\bK_{1,3}$:
    \begin{itemize}
        \item $\bD$ is $\{\bF,\vec{\bP}_5,\bP_5^{\leftarrow\leftarrow\to\to},\bP_5^{\to\to\leftarrow\leftarrow}\}$-subgraph-free,
        \item $\bD$ is $\vec{\bP}_4$-free, and
        \item $d^+(v) + d^-(v) \le 3$ for every $v\in D$.
    \end{itemize}
\end{lemma}
\begin{proof}
We propose a polynomial-time reduction from positive 1-IN-3 SAT to $\CSP(\vec{\bC}_3^+)$
mapping an instance $\bI$ to a digraph $\bD$ satisfying the statements above.
Let $\bI$ be an instance $(x_1\lor y_1 \lor z_1)\land \dots \land (x_m \lor y_m \lor z_m)$
   where all variables are positive. Construct $\bD$ as follows:
   \begin{itemize}
       \item For each clause $(x_i\lor y_i \lor z_i)$ introduce  a fresh copy $\bG_i$ of $\bG$
       (depicted in Figure~\ref{fig:C3+-gadget}) with distinguished vertices $x_i,y_i,z_i$, and
       \item for each variable $v$ that occurs $n_v$ times, construct an undirected cycle with exactly
       $n_v$ vertices. Now substitute each edge $pq$ in this cycle for a gadget of two edges on
       three vertices given by the back-and-forward formation $(r_{p,q},p),(r_{p,q},q)$, where $r_{p,q}$ is a new vertex. 
       Notice that the resulting digraph is an oriented cycle $\bC_v$ where the $n_v$ vertices from the original
       undirected cycle correspond to sinks in $\bC_v$.
       Now, when variable $v$ appears in the $i$-th clause, identify the vertex $v$ in the clause gadget $\bG_i$ with a unique sink  
       the oriented cycle $\bC_v$.
   \end{itemize}
   Notice that $\exists w E(w,x) \wedge E(w,y)$ defines an equivalence relation in $\vec{\bC}_3^+$
   with two classes $\{1,3\}$ and $\{2\}$. Hence, for each $a\in\{1,2,3\}$, any homomorphism $f\colon
   \bG\to \vec{\bC}_3^+$ satisfies that $f(a') = 2$ if and only if $f(a) = 2$. Moreover, it follows from 
   the same argument that a sink in the cycle $\bC_v$ is mapped to $2$ if and only if all sinks in $\bC_v$ are mapped to $2$.
   Also notice that any homomorphism from the directed $3$-cycle to $\vec{\bC}_3^+$ must map exactly one vertex to $2$. 
   Hence, any homomorphism $g\colon\bD\to \vec{\bC}_3^+$ satisfies that exactly one of
   the vertices $x_i,y_i,z_i$ is mapped to $2$. This yields a solution to the positive
   1-IN-3 SAT instance $\bI$ by assigning for each $i\in[m]$ and $a\in\{x,y,z\}$
   the value $1$ if $g(a_i) = 2$, and $a_i := 0$ otherwise. The converse implication
   (if $\bI$ is a yes-instance then $\bD\to \vec{\bC}_3^+$) follows similarly by noticing
   that for each $a\in \{x,y,z\}$ there is a homomorphism $f\colon \bG\to \vec{\bC}_3^+$
   mapping $a$ to $2$ and $f(b) = 1$ for $b\in\{x,y,z\}\setminus\{a\}$ (in Figure~\ref{fig:C3+-gadget}
   we describe such a homomorphism for $a = x$).

   Finally, it is straightforward to notice that every digraph $\bD$ in the image of the reduction
   contains no path on five vertices, and no induced path on four vertices. It is also clear
   that $\bD$ is $\{\bP_5^{\leftarrow\leftarrow\to\to},\bP_5^{\to\to\leftarrow\leftarrow}\}$-subgraph-free,
   and  that $d^+(u) + d^-(u) \le 3$ for every $u\in D$. Now,  notice that if $d^+(u) + d^-(u) = 3$,
   then $u$ belongs to some directed $3$-cycle from a gadget $\bG_i$ or is a source in some cycle
   $\bC_v$. In the former case $d^-(v) = 2$ and in the latter $d^-(v) = 0$, and thus $\bD$ is
   $\bF$-subgraph-free whenever $\bF$ is the orientation of $\bK_{1,3}$ where the center
   vertex is a sink, or when it has out-degree two. The cases when $\bF$ is on of the remaining
   two possible orientations of $\bK_{1,3}$ simply follows by considering the reduction
   that maps a digraph $\bD$ to the digraph obtained from $\bD$ by reversing the orientation of each edge of $\bD$. 
\end{proof}

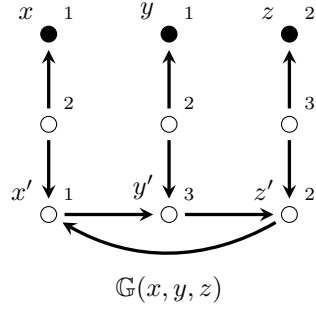
\begin{figure}[ht!]
\centering
\begin{tikzpicture}

    \begin{scope}[scale = 0.8]
    \node [vertex, fill = black, label =135:{$x$},label =45:{$\scriptstyle 1$}] (x) at (-2,3) {};
    \node [vertex, fill = black, label = 135:{$y$}, label =45:{$\scriptstyle 1$}] (y) at (0,3) {};
    \node [vertex, fill = black, label = 135:{$z$}, label =45:{$\scriptstyle 2$}] (z) at (2,3) {};
    \node [vertex, label =45:{$\scriptstyle 2$}] (x2) at (-2,1.5) {};
    \node [vertex, label =45:{$\scriptstyle 2$}] (y2) at (0,1.5) {};
    \node [vertex, label =45:{$\scriptstyle 3$}] (z2) at (2,1.5) {};
    \node [vertex, label =135:{$x'$}, label =45:{$\scriptstyle 1$}] (x3) at (-2,0) {};
    \node [vertex, label =135:{$y'$}, label =45:{$\scriptstyle 3$}] (y3) at (0,0) {};
    \node [vertex, label =135:{$z'$}, label =45:{$\scriptstyle 2$}] (z3) at (2,0) {};
    \node at (0,-1.25){$\bG(x,y,z)$};
   
    \foreach \from/\to in {x2/x, y2/y, z2/z, x2/x3, y2/y3, z2/z3, x3/y3, y3/z3}     
    \draw [arc] (\from) to (\to);
    
    \draw [arc] (z3) to [bend left] (x3);
  \end{scope}

\end{tikzpicture}
\caption{
A depiction of the gadget reduction $\bI\mapsto \bD$ from positive 1-IN-3 SAT to
$\CSP(\vec{\bC}_3^+)$  applied to a clause $(x\lor y \lor z)$ of the instance $\bI$ to
1-IN-3 SAT. The numbers indicate a function that defined a homomorphism 
$f\colon \bG\to \vec{\bC}_3^+$.}
\label{fig:C3+-gadget}
\end{figure}

\blue{Now we show that $\CSP(\vec{\bC}_3^+)$ can be solved in polynomial time when the
input digraph $\bD$ contains no induced directed path on three vertices. On a given
$\vec{\bP}_3$-free digraph $\bD$, the algorithm works as follows. At each
step there are three sets $X_1,X_2,X_3\subseteq D$ such that the mapping $x\mapsto i$ if
$x\in X_i$ defines a partial homomorphism from $\bD$ to $\vec{\bC}_3^+$. We extend
these sets to $X_1',X_2',X_3'$ in such a way that the partial homomorphism
defined by $X_1,X_2,X_3$ extends to a homomorphism from $\bD\to \vec{\bC}_3^+$ if
and only if the partial homomorphism defined by $X_1',X_2',X_3'$ extends to 
a homomorphism from $\bD\to \vec{\bC}_3^+$.}

\begin{lemma}\label{lem:alg-C3+}
    \blue{$\CSP(\vec{\bC}_3^+)$ can be solved in quadratic time when the input is restricted to $\vec{\bP}_3$-free
    digraphs.}
\end{lemma}
\begin{proof}
    \blue{We assume without loss of generality that $\bD$ is a $\vec{\bC}_3$-free weakly connected digraph.
    First notice that $\bD\to \vec{\bC}_3^+$ if and only if there is a homomorphism $f\colon \bD\to
    \vec{\bC}_3$ such that $f^{-1}(2)\neq \varnothing$. As previously mentioned, the idea of the algorithm
    is to deterministically extend a partial homomorphism $f$ to a partial homomorphism $f'$ in such a way
    that $f$ can be extended to a homomorphism $\bD\to \vec{\bC}_3^+$ if an only if $f'$ can be extended
    as well. Hence, by the previous observation it sufficed to run  this routine over all partial homomorphisms
    defined on one vertex $v$ where $v\mapsto 2$, and if some homomorphism is found, then we accept
    $\bD$, and otherwise we reject $\bD$.}

    \blue{Let $\bD$ be  a finite digraph, $X_1,X_2,X_3$ three disjoint vertex subsets
    of $D$ such that $X_1\cup X_2 \cup X_3\neq \varnothing$, and $x\mapsto i$
    if $x\in X_i$ defines a partial homomorphism $f\colon\bD[X_1\cup X_2\cup X_3]\to \vec{\bC}_3^+$.
    Now, consider any vertex $v\in D\setminus (X_1\cup X_2\cup X_3)$ that belongs
    to the in- or out-neighbourhood of $X_1\cup X_2\cup X_3$. We define $X_1',X_2',X_3'$
    according to the following case distinction.}
\begin{enumerate}
    \item If $v\in N^+(X_1)$ or $v\in N^-(X_2)$, then $X_1': = X_1$, $X_2':=X_2$, and $X_3':= X_3\cup \{v\}$;
    \item If $v\in N^-(X_1)$ or $v\in N^+(X_3)$, then $X_1': = X_1$, $X_2':=X_2\cup\{v\}$, and $X_3':= X_3$;
    \item If $v\in N^+(X_2)$, we consider the following subcases,
        \begin{enumerate}
            \item if there is a symmetric edge incident in $v$, then $X_1': = X_1$, $X_2':=X_2$, and $X_3':= X_3\cup \{v\}$;
            \item otherwise, if $v$ has an out-neighbour, then $X_1':= X_1\cup \{v\}$, $X_2':=X_2$,  and $X_3':= X_3$;
            \item if neither of the above hold, then $X_1': = X_1$, $X_2':=X_2$, and $X_3':= X_3\cup \{v\}$;
        \end{enumerate}
    \item If $v\in N^-(X_3)$, we consider the following subcases,
        \begin{enumerate}
            \item if there is a symmetric edge incident in $v$, then $X_1': = X_1$, $X_2':=X_2\cup\{v\}$, and $X_3':= X_3$;
            \item otherwise, if $v$ has an in-neighbour, then $X_1':= X_1\cup \{v\}$, $X_2':=X_2$,  and $X_3':= X_3$;
            \item if neither of the above hold, then $X_1': = X_1$, $X_2':=X_2 \cup\{v\}$, and $X_3':= X_3$.
        \end{enumerate}
\end{enumerate}

\blue{Let $f'\colon\bD[X_1'\cup X_2'\cup X_3']\to \vec{\bC}_3^+$ be the partial mapping $x \mapsto i$
if $x \in X_i'$. We prove that if $f$ extends to a homomorphism from $\bD$ to $\vec{\bC}_3^+$, then
so does $f'$. This is clearly the case in cases 1 and 2, and also in 3a and 4a. We now consider the
case 3b. Let $x\in X_2$ and $y\in D$ such that $(x,v), (v,y)\in E(\bD)$. Since $v$ is not incident
in any symmetric edge, then $x$ and $y$ must be adjacent vertices in $\bD$, i.e.,  $(x,y)\in E(\bD)$
or $(y,x)\in E(\bD)$ (or both). In either case, there is a unique way of extending the
partial homomorphism $x\mapsto 2$ to the subbgraph with vertex set $\{x,y,v\}$, and in this unique
extension it is the case that $v$ is mapped to $1$. Hence, $f$ extends to a homomorphism if and only
if $f'$ extends to a homomorphism $\bD\to \vec{\bC}_3^+$. Now, consider the case 3c and let
$x\in X_2$ be an in-neighbour of $v$. Since $v$ no symmetric edge is incident in $v$, and $v$ has no out-neighbours,
it follows that for any vertex $y$ adjacent to $v$, there is a walk of the form $\rightarrow\leftarrow$
connecting $y$ and $x$. Notice that if $i\in\{1,2,3\}$ is connected to $2$ by such a walk in $\vec{\bC}_3^+$, then
$i\in\{1,2\}$. This implies  that if some homomorphism $g\colon \bD\to \vec{\bC}_3^+$ mapping $x$ to $2$,
then the image of the neighbourhood of $v$ is contained in $\{1,2\}$. Since $v$ has in-neighbours but no
out-neighbours, and $1$ and $2$ are in-neighbours of $3$ in $\vec{\bC}_3^+$, it follows that the mapping
$g'\colon \bD\to \vec{\bC}_3^+$ defined by $g'(v) = 3$ and $g'(u) = g(u)$ if $u\neq v$ is a also a homomorphism
from $\bD\to \vec{\bC}_3^+$. Therefore, if $f$ extends to a homomorphism $g\colon \bD\to \vec{\bC}_3^+$, 
then $f'$ extends to a homomorphism $g'\colon \bD\to \vec{\bC}_3^+$.
The cases 4b and 4c follows with symmetric arguments.}

\blue{
Hence, on a given weakly connected digraph $\bD$,  a quadratic time algorithm works as follows.
For a vertex $v\in D$ define $X_1 := \varnothing$, $X_2:= \{v\}$, and $X_3:=\varnothing$,
and perform the subroutine above until either $X_1'\cup X_2'\cup X_3'$ does not define a
partial homomorphism, and in this case choose a new vertex $u\in D$; or 
$X_1\cup X_2 \cup X_3 = D$, and in this case accept $\bD$. If after running the subroutine
on every vertex $v\in D$ we arrive to some sets $X_1',X_2',X_3'$ that do not define a partial homomorphism,
reject $\bD$.
}
\end{proof}

\begin{lemma}\label{lem:C3++-P3-free}
    \blue{$\CSP(\vec{\bC}_3^{++})$ is $\NP$-hard even for a
    $\{\vec{\bP}_3, \vec{\bP}_3^{\leftarrow\rightarrow},\vec{\bP}_3^{\rightarrow\leftarrow}\}$-free digraphs.}
\end{lemma}
\begin{proof}
We follow the proof of Lemma~\ref{lem:C3+-P4-free}, noting that $\exists w .E(w,x) \wedge E(x,w) \wedge E(w,y) \wedge E(y,w)$
defines an equivalence relation with two classes $\{1,2\}$ and $\{3\}$. The gadget we use is
depicted in  Figure~\ref{fig:gadget-C3++-P3-free}, with the same reduction from 1-IN-3 SAT.
\end{proof}

\begin{figure}[ht!]
\centering
\begin{tikzpicture}

    \begin{scope}[scale = 0.8]
    \node [vertex, fill = black, label =135:{$x$},label =45:{$\scriptstyle 1$}] (x) at (-2,3) {};
    \node [vertex, fill = black, label = 135:{$y$}, label =45:{$\scriptstyle 2$}] (y) at (0,3) {};
    \node [vertex, fill = black, label = 135:{$z$}, label =45:{$\scriptstyle 1$}] (z) at (2,3) {};
    \node [vertex, label =right:{$\scriptstyle 2$}] (x2) at (-2,1.5) {};
    \node [vertex, label =right:{$\scriptstyle 1$}] (y2) at (0,1.5) {};
    \node [vertex, label =right:{$\scriptstyle 2$}] (z2) at (2,1.5) {};
    \node [vertex, label =135:{$x'$}, label =25:{$\scriptstyle 3$}] (x3) at (-2,0) {};
    \node [vertex, label =135:{$y'$}, label =25:{$\scriptstyle 2$}] (y3) at (0,0) {};
    \node [vertex, label =135:{$z'$}, label =25:{$\scriptstyle 1$}] (z3) at (2,0) {};
    \node at (0,-1.25){$\bG(x,y,z)$};
   
    \foreach \from/\to in {x3/y3, y3/z3}     
    \draw [arc] (\from) to (\to);

    \foreach \from/\to in {x2/x, x/x2, y/y2, y2/y, z/z2, z2/z, x3/x2, x2/x3, y3/y2, y2/y3, z3/z2, z2/z3} 
    \draw [arc] (\from) to [bend left = 20] (\to);
    \draw [arc] (z3) to [bend left] (x3);
  \end{scope}

\end{tikzpicture}
\caption{
A depiction of the gadget reduction $\bI\mapsto \bD$ from positive 1-IN-3 SAT to
$\CSP(\vec{\bC}_3^{++})$  applied to a clause $(x\lor y \lor z)$ of the instance $\bI$ to
1-IN-3 SAT. The numbers indicate a function that defined a homomorphism 
$f\colon \bG\to \vec{\bC}_3^{++}$.}
\label{fig:gadget-C3++-P3-free}
\end{figure}

The following statement now follows from Lemmas~\ref{lem:C3+-P4-free},~\ref{lem:alg-C3+}, and~\ref{lem:C3++-P3-free}.

\begin{theorem}\label{thm:Pk-3vertex-classification}
    \blue{The following statements hold for every positive integer $k\ge 2$.
    \begin{itemize}
        \item $\CSP(\vec{\bC}_3^+)$ is in solvable in quadratic time for $\vec{\bP}_k$-free digraphs
        if $k\le 3$, and $\NP$-hard even for $\vec{\bP}_k$-free digraphs if $k\ge 4$, 
        \item $\CSP(\vec{\bC}_3^{++})$ is solvable in linear time for $\vec{\bP}_k$-free digraphs
        if $k = 2$, and  $\NP$-hard even for $\vec{\bP}_k$-free digraphs if $k\ge 3$, and
        \item $\CSP(\bK_3)$ is $\NP$-hard even for $\vec{\bP}_k$-free digraphs.
    \end{itemize}}
\end{theorem}

\section{A family of smooth tournaments}
\label{sect:tournaments}

In this section we answer Question~\ref{qst:Pk} for a natural family of tournaments 
smooth tournaments $\bT\bC_n$. 
Given a positive integer $n$, we denote by $\bT\bC_n$
the tournament obtained from $\bT_n$ be reversing the edge from the source to the sink (see
Figure~\ref{fig:TCn}). In particular, 
$\bT\bC_2 \cong \bT_2$, and $\bT\bC_3 \cong \vec{\bC}_3$, so $\CSP(\bT\bC_n)$ is
polynomial-time solvable for $n \le 3$, and $\NP$-complete for $n\ge 4$
(see, e.g.,~\cite{bangjensenSIDMA1}).

\begin{figure}[ht!]
\centering
\begin{tikzpicture}

    \begin{scope}[scale = 0.8]
    \node [vertex, label =below:{$1$}] (v) at (-4,0) {};
    \node [vertex,  label = below:{$2$}] (a) at (-2.5,0) {};
    \node [vertex,  label = below:{$3$}] (b) at (-1,0) {};
    \node [vertex,  label = below:{$n-2$}] (c) at (1,0) {};
    \node [vertex, label = below:{$n-1$}] (d) at (2.5,0) {};
    \node [vertex, label = below:{$n$}] (e) at (4,0) {};
    \node at (0,-1.5) {$\bT\bC_n$};
    \node at (0,0) {$\dots$};

    \foreach \from/\to in {v/a, a/b, c/d, d/e}     
    \draw [arc] (\from) to (\to);

    \foreach \from/\to in {a/d, a/c, b/d}     
    \draw [arc] (\from) to [bend left] (\to);
    \foreach \from/\to in {v/b, c/e}     
    \draw [arc] (\from) to [bend left = 20] (\to);
    \draw [arc] (e) to [bend right = 40] (v);
  \end{scope}

\end{tikzpicture}
\caption{
A depiction of the $\bT\bC_n$ --- for a cleaner picture we omit the edges $(1,n-2), (1,n-1), (2,n),$ and $(3,n)$.}
\label{fig:TCn}
\end{figure}
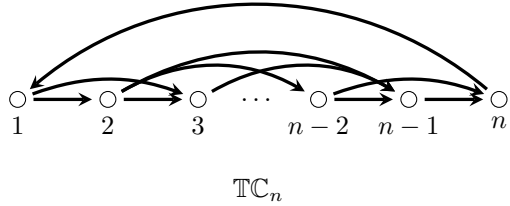

Since $\bT\bC_n$ is a hereditary hard digraph (Theorem~\ref{thm:smooth-hereditay}) and $\bT_n\not\to \bT\bC_n$,
it follows that $\RCSP(\bT\bC_n,\bT_n)$ is $\NP$-hard (Theorem~\ref{thm:brewster}). Equivalently, $\CSP(\bT\bC_n)$ is
$\NP$-hard for digraphs with no directed walk on $n+1$ vertices, and since $\bT_{n-1}\to \bT\bC_n$, $\CSP(\bT\bC_n)$
is polynomial-time solvable for digraphs with no directed walk on $n$ vertices. However, is we only forbid
$\vec{\bP}_k$ as a subgraph (and not homomorphically), it turns out 
that $\CSP(\bT\bC_n)$ is $\NP$-hard even for $\vec{\bP}_5$-subgraph-free digraphs.

\subsection{$\vec{\bP}_k$-subgraph-free digraphs}

For these hardness results we consider the gadget reduction depicted in Figure~\ref{fig:gen-pp-def} and described in
the proof of the following lemma. 

\begin{figure}[ht!]
\centering
\begin{tikzpicture}

    \begin{scope}[scale = 0.8]
    \node [vertex, fill = black, label =135:{$x$},label =45:{$\scriptstyle 1$}] (x) at (-2,3) {};
    \node [vertex, fill = black, label = 135:{$y$}, label =45:{$\scriptstyle n$}] (y) at (0,3) {};
    \node [vertex, fill = black, label = 135:{$z$}, label =45:{$\scriptstyle n$}] (z) at (2,3) {};
    \node [vertex, label =45:{$\scriptstyle n$}] (x2) at (-2,1.5) {};
    \node [vertex, label =45:{$\scriptstyle 2$}] (y2) at (0,1.5) {};
    \node [vertex, label =45:{$\scriptstyle 2$}] (z2) at (2,1.5) {};
    \node [vertex, label =135:{$x'$}, label =45:{$\scriptstyle 1$}] (x3) at (-2,0) {};
    \node [vertex, label =135:{$y'$}, label =45:{$\scriptstyle 3$}] (y3) at (0,0) {};
    \node [vertex, label =135:{$z'$}, label =45:{$\scriptstyle n$}] (z3) at (2,0) {};
    \node at (0,-1.25){$\bG(x,y,z)$};
   
    \foreach \from/\to in {x2/x, y2/y, z2/z, x2/x3, y2/y3, z2/z3, x3/y3, y3/z3}     
    \draw [arc] (\from) to (\to);
    
    \draw [arc] (z3) to [bend left] (x3);
  \end{scope}

\end{tikzpicture}
\caption{
A depiction of the gadget reduction $\bI\mapsto \bD$ from positive 1-IN-3 SAT to
$\CSP(\bT\bC_n)$  applied to a clause $(x\lor y \lor z)$ of the instance $\bI$ to
1-IN-3 SAT. The numbers indicate a function that defined a homomorphism 
$f\colon \bG\to \bT\bC_n$ whenever $n\ge 4$.}
\label{fig:gen-pp-def}
\end{figure}

\begin{lemma}\label{lem:TCkxTTk}
    For every positive integer $n\ge 4$, $\CSP(\bT\bC_n)$ is $\NP$-hard even when the input
    $\bD$ satisfies the following conditions:
    \begin{itemize}
        \item $\bD$ is \blue{$\{\bF,\vec{\bP}_5,\bP_5^{\leftarrow\leftarrow\to\to},\bP_5^{\to\to\leftarrow\leftarrow}\}$-subgraph-free,} 
        \item $\bD$ is $\vec{\bP}_4$-free, and
        \item $d^+(v) + d^-(v) \le 3$ for every $v\in D$.
    \end{itemize}
\end{lemma}
\begin{proof}
    We use the same reduction from 1-IN-3-SAT to $\CSP(\vec{\bC}_3^+)$ from Lemma~\ref{lem:C3+-P4-free}:
    \begin{itemize}
       \item For each clause $(x_i\lor y_i \lor z_i)$ introduce  a fresh copy $\bG_i$ of $\bG$
       (depicted in Figure~\ref{fig:gen-pp-def}) with distinguished vertices $x_i,y_i,z_i$, and
       \item for each variable $v$ that occurs $n_v$ times, construct \blue{an undirected cycle with exactly}
       $n_v$ \blue{vertices}. Now substitute each edge $pq$ in this \blue{cycle} for a gadget of two edges on
       three vertices given by the back-and-forward formation $(r_{p,q},p),(r_{p,q},q)$, where $r_{p,q}$ is a new vertex. 
       Notice that the resulting digraph is an oriented cycle $\bC_v$ where the $n_v$ vertices from the original
       undirected cycle correspond to sinks in $\bC_v$. Now, when variable $v$ appears in the $i$-th clause, identify
       the vertex $v$ in the clause gadget $\bG_i$ with a unique sink the oriented cycle $\bC_v$.
   \end{itemize}
    Notice that any homomorphism $f\colon \bG\to \bT\bC_n$ satisfies that $f(a') = 1$, 
   if and only if $f(a) = 1$ for each $a\in\{x,y,z\}$. Indeed, this is independent of the cycle in this gadget,
   it applies already to the back-and-forward formation $(r_{a,a'},a),(r_{a,a'},a')$. It thus follows that
    a sink in the cycle $\bC_v$ is mapped to $1$ if and only if all sinks in $\bC_v$ are mapped to $1$.
    Also notice that any homomorphism from the directed $3$-cycle to $\bT\bC_n$ must map exactly one vertex to $1$. 
   Hence, any homomorphism $g\colon\bD\to \bT\bC_n$ satisfies that exactly one of
   the vertices $x_i,y_i,z_i$ is mapped to $1$. This yields a solution to the positive
   1-IN-3 SAT instance $\bI$ by assigning for each $i\in[m]$ and $a\in\{x,y,z\}$
   the value $1$ if $g(a_i) = 1$, and $a_i := 0$ otherwise. The converse implication
   (if $\bI$ is a yes-instance then $\bD\to \bT\bC_n$) follows similarly by noticing
   that for each $a\in \{x,y,z\}$ there is a homomorphism $f\colon \bG\to \bT\bC_n$
   mapping $a$ to $1$ and $f(b) = n$ for $b\in\{x,y,z\}\setminus\{a\}$ (in Figure~\ref{fig:gen-pp-def}
   we describe such a homomorphism for $a = x$). 

   Finally, the fact that $\bD$ satisfies the structural restrictions follows with the same arguments
   as in the proof of Lemma~\ref{lem:C3+-P4-free}.
\end{proof}

\blue{We now argue that $\CSP(\bT\bC_n)$ is polynomial-time solvable
when the input is restricted to $\vec{\bP_4}$-subgraph-free digraphs, and together
with Lemma~\ref{lem:TCkxTTk} we obtain a complexity classification for
these CSPs restricted to $\vec{\bP}_k$-subgraph-free digraphs.}

\begin{lemma}\label{lem:oriented-P4-subgraph-free}
    Consider a (possibly infinite) digraph $\bH$. If $\bT_3\to \bH$, then $\CSP(\bH)$ is in
    $\cP$ for $\vec{\bP}_4$-subgraph-free oriented graphs. In particular, 
    if $\bH$ is an oriented graph and $\bT_3\to \bH$, then $\CSP(\bH)$ is in
    $\cP$ for $\vec{\bP}_4$-subgraph-free digraphs.
\end{lemma}
\begin{proof}
    The claim is trivial when $\bH$ contains a loop, so we assume that 
    $\bH$ is a loopless digraph. Let $\bD$ be an oriented graph and without
    loss of generality assume it
    is a weakly connected digraph. First check if $\bD$ contains a directed
    $3$-cycle, and if yes, then by Remark~\ref{rmk:3-cycle} we know that $D$ has 
    exactly three vertices
    (at this step the algorithm would say yes if $\bD$ is one the subdigraphs
    of $\bH$ on three vertices). If not, notice that $\vec{\bP}_4\not\to\bD$
    if and only if $\bD$ contains no loops.
    Hence, the algorithm accepts whenever $\bD$ contains no loops, and rejects
    otherwise. The ``in particular''
    statement follows because if the input $\bD$ is not an oriented graph
    but $\bH$ is, then $\bD\not\to \bH$, hence it suffices to solve $\CSP(\bH)$
    for $\vec{\bP}_4$-subgraph-free oriented graphs.
\end{proof}

\begin{theorem}\label{thm:TCn-Pk-subgraph-classification}
    For every pair of positive integers $n$ and $k$ the following statements hold.
    \begin{itemize}
        \item If $n \le 3$ or $k\le 4$, then $\CSP(\bT\bC_n)$ is in $\cP$ for $\vec{\bP}_k$-subgraph-free
        digraphs.
        \item If $n \ge 4$ and $k \ge 5$, then $\CSP(\bT\bC_n)$ is $\NP$-hard even for
        $\vec{\bP}_k$-subgraph-free digraphs.
    \end{itemize}
\end{theorem}
\begin{proof}
    The second statement follows via Lemma~\ref{lem:TCkxTTk}, and the case
    $n \ge 3$ holds because $\bT\bC_2$ and $\bT\bC_3$ are transitive tournaments.
    Finally, the case $k\le 4$ and $n\ge 4$ follows from Lemma~\ref{lem:oriented-P4-subgraph-free}
    because each $\bT_n$ is an oriented graph that contains a transitive tournament
    three vertices whenever $n\ge 3$.
\end{proof}

\subsection{$\vec{\bP}_k$-free digraphs}

In this subsection we prove a structural result (Theorem~\ref{thm:TCn-min-obs})
asserting that there is a finite set of digraphs $\calF$ such that a $\vec{\bP}_3$-free
digraphs $\bD$ belongs to $\CSP(\bT\bC_n)$ if and only if $\bD$ is $\calF$-free. 
We then use this result to  propose a complexity classification for
these CSPs restricted to $\vec{\bP}_k$-free digraphs.

We begin by showing that there are finitely many minimal $\vec{\bP}_3$-free digraphs
that do not homomorphically map to $\CSP(\bT\bC_4)$. We depict the four non-isomorphic
tournaments on four vertices in Figure~\ref{fig:four-vertices}.

\begin{lemma}\label{lem:TC4-min-obs}
    The following statements are equivalent for a $\vec{\bP}_3$-free digraph $\bD$.
    \begin{itemize}
        \item $\bD\to \bT\bC_4$, and
        \item $\bD$ is a $\{\bT_4,\bT_4^a, \bT_4^b\}$-free loopless oriented graph with no tournament
        on five vertices.
    \end{itemize}
\end{lemma}
\begin{proof}
    The first itemized statement clearly implies the second one.
    To prove the converse implication we proceed by case distinction over the
    largest tournament in $\bD$, and without loss of generality we assume
    that $\bD$ is a weakly connected digraph. 
    \begin{itemize}
        \item If $\bD$ is a $\vec{\bP}_3$-free loopless oriented graph with \textbf{no tournament
        on three vertices}, then $\vec{\bP}_3\not\to \bD$ so $\bD\to \bT_2$ (Observation~\ref{obs:P-TT}),
        and thus $\bD\to \bT\bC_4$.
        \item If $\bD$ contains \textbf{no tournament on four vertices and no ${\vec \bC_3}$}, 
        then there is no homomorphism $\vec{\bP}_4\to \bD$, so $\bD\to \bT_3$ (Observation~\ref{obs:P-TT}),
        which implies that $\bD\to \bT\bC_4$.
        \item Now, suppose that $\bD$ contains \textbf{no tournament on four vertices and a ${\vec \bC_3}$} with
        vertices $d_1,d_2,d_3$ and edges $(d_1,d_2),(d_2,d_3),(d_3,d_1)$.
        Notice that if $d$ is an out-neighbour of $d_2$, then $d$ is an in-neighbour of $d_1$.
        Indeed, since $\bD$ is $\vec{\bP}_3$-free, it must be the case that $d$ is a neighbour
        of $d_1$, and if $(d_1,d)\in E(\bD)$ it would also be the case that $d$ is a neighbour
        of $d_3$ contradicting the choice of $\bD$. Hence, every out-neighbour of $d_2$ is
        an in-neighbour of $d_1$, and symmetrically, every in-neighbour of $d_2$ is an out-neighbour
        of $d_3$. Therefore, the mapping $f$ defined by $d\mapsto 1$ if $d\in N^-(d_2)$, $d\mapsto 2$
        if $d\in N^-(d_3)$, and $d\mapsto 3$ if $d\in N^-(d_1)$ is a partial homomorphism from 
        $\bD$ to $\vec{\bC}_3$. The fact that $f$ is a homomorphism also follows from the previous
        observation and the assumption that $\bD$ is weakly connected.
        \item Finally, suppose that $\bD$ \textbf{contains a tournament on four vertices} $d_1,d_2,d_3,d_4$.
        Since $\bD$ is $\{\bT_4,\bT_4^a, \bT_4^b\}$-free, we assume without loss of generality that the mapping $i\mapsto d_i$
        is an embedding of $\bT\bC_4$ into $\bD$. Let $d$ be an out-neighbour of $d_4$. Since
        $(d_2,d_4),(d_3,d_4)\in E(\bD)$, it follows that $d$ is a neighbour of $d_2$ and of $d_3$.
        Notice that if $(d_2,d)\in E(\bD)$ or $(d_3,d)\in E(\bD)$, then the $\vec{\bP}_3$-freeness of
        $\bD$ implies that $d$ is also a neighbour of $d_1$ contradicting the assumption that
        $\bD$ contains no tournament on five vertices. Hence, every out-neighbour of $d$ is
        an in-neighbour of $d_2$ and of $d_3$, and it is not adjacent to $d_1$. With symmetric arguments
        it follows that every in-neighbour of $d_1$ is an out-neighbour of $d_2$ and of $d_3$, and
        it is not adjacent to $d_4$. We define $X_1:=N^+(d_4)$, $X_4:=N^-(d_1)$, and $Y:= N^+(d_1)$.
        We now claim that $Y = N^+(d_1) = N^-(d_4)$ and that $(X_1,Y,X_4)$ is a partition of the vertex set $D$.
        Again, using the fact that $\bD$ is $\vec{\bP}_3$-free one can notice that every $d\in N^+(d_1)$
        is a neighbour of $d_4$ (because $(d_4,d_1)\in E(\bD)$), and since every out-neighbour of
        $d_4$ is not adjacent to $d_1$ (by the previous arguments), it follows that $d\in  N^-(d_4)$. 
        Hence, $N^+(d_1)\subseteq N^-(d_4)$, and with symmetric arguments we conclude that
        $N^+(d_1) = N^-(d_4)$. Since $\bD$ is an oriented graph, the sets $X_1$,  $N^+(d_1)$, and
        $X_4$ are disjoint, and -- anticipating a contradiction -- suppose there is a vertex $d\in D$
        not adjacent to $d_1$ nor $d_4$. \blue{Moreover,} since $\bD$ is weakly connected, \blue{we
        choose $d$ so that it} does not belong to $X_1\cup N^+(d_1) \cup X_4$ \blue{and it}
        is adjacent to some vertex $c\in X_1\cup N^+(d_1) \cup X_4$. To arrive to a contradiction
        it suffices to notice that $d$ is a neighbour of $d_1$ or of $d_4$ --- this implies
        that $d\in N^+(d_1) \cup N^-(d_1) \cup N^+(d_4)\cup N^-(d_4) = X_1\cup N^+(d_1)\cup X_4$.
        To do so one can proceed with a case distinction over which set $c$ belongs to (from $X_1$, 
        $N^+(d_1)$, or \blue{$X_4)$}),
        and whether $(d,c)\in E(\bD)$ or $(c,d)\in E(\bD)$. We give a sketch and leave the missing
        details to the reader: 
        in case of $c\in N^+(d_1)$, $d$ ends up adjacent to $d_4$ if $(d,c)\in E(\bD)$,
        and adjacent to $d_1$ if $(c,d)\in E(\bD)$; in case of $c\in X_1$, $d$ ends up adjacent to $d_4$ 
        if $(c,d)\in E(\bD)$, and adjacent to some $c'\in N^+(d_1)$, if $(c,d)\in E(\bD)$, hence we fall in the first case;
        finally, the case when $c\in X_4$ follows with symmetric arguments to the case of $c\in X_1$.
        
        Now, notice that $N^+(d_4)$ cannot contain three vertices
        $c_1,c_2,c_3$ that induces either a directed three cycle, or a transitive tournament;
        otherwise, $d_1,c_1,c_2,c_3$ induce a tournament on four vertices not isomorphic to $\bT\bC_4$
        contradicting the choice of $\bD$. From this observation together with the fact that
        $\bD$ is $\vec{\bP}_3$-free we conclude that the subdigraph of $\bD$ induced by
        $N^+(d_1)$ does not contain an oriented walk on three vertices. Hence, by Observation~\ref{obs:P-TT},
        there \blue{is a homomorphism $h\colon \bD[N^+(d_1)]\to \bT\bT_2$ from the subdigraph of $\bD$
        induced by $N^+(d_1)$ to the transitive tournament on two vertices, i.e., there}\todo{added a sentence here}
        is a partition $(X_2,X_3)$ of $N^+(d_1)$ such that for $c,d\in X_2,X_3$ if
        $(c,d)\in E(\bD)$, then $c\in X_2$ and $d\in X_3$. We finally conclude that
        the function $f$ defined by $d\mapsto i$ if $d\in X_i$ for $i\in\{1,2,3,4\}$ defines
        a homomorphism $f\colon \bD\to \bT\bC_4$.
    \end{itemize}    
\end{proof}

\begin{figure}[ht!]
\centering
\begin{tikzpicture}

  \begin{scope}[xshift = -4.5cm, scale = 0.8]
    \node [vertex] (1) at (-1,2) {};
    \node [vertex] (2) at (1,2) {};
    \node [vertex] (4) at (-1,0) {};
    \node [vertex] (3) at (1,0) {};
    \node (L1) at (0,-1) {$\bT_4$};
      
    \foreach \from/\to in {1/2, 1/3, 1/4, 2/3, 2/4, 3/4}     
    \draw [arc] (\from) to (\to);
  \end{scope}

  \begin{scope}[xshift = -1.5cm, scale = 0.8]
    \node [vertex] (1) at (-1,2) {};
    \node [vertex] (2) at (1,2) {};
    \node [vertex] (4) at (-1,0) {};
    \node [vertex] (3) at (1,0) {};
    \node (L1) at (0,-1) {$\bT\bC_4$};
      
    \foreach \from/\to in {1/2, 1/3, 4/1, 2/3, 2/4, 3/4}     
    \draw [arc] (\from) to (\to);
  \end{scope}

  \begin{scope}[xshift = 1.5cm, scale = 0.8]
    \node [vertex] (1) at (-1,2) {};
    \node [vertex] (2) at (1,2) {};
    \node [vertex] (4) at (-1,0) {};
    \node [vertex] (3) at (1,0) {};
    \node (L1) at (0,-1) {$\bT_4^a$};
    
    \foreach \from/\to in {1/2, 1/3, 1/4, 2/3, 4/2, 3/4}     
    \draw [arc] (\from) to (\to);
  \end{scope}

  \begin{scope}[xshift = 4.5cm, scale = 0.8]
    \node [vertex] (1) at (-1,2) {};
    \node [vertex] (2) at (1,2) {};
    \node [vertex] (4) at (-1,0) {};
    \node [vertex] (3) at (1,0) {};
    \node (L1) at (0,-1) {$\bT_4^b$};
      
    \foreach \from/\to in {2/1/, 3/1, 4/1, 2/3, 4/2, 3/4}     
    \draw [arc] (\from) to (\to);
  \end{scope}

\end{tikzpicture}
\caption{The four non-isomorphic oriented tournaments on $4$ vertices}
\label{fig:four-vertices}
\end{figure}
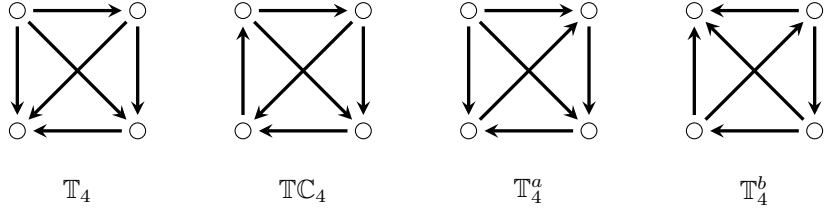

\blue{
For the remaining of this section, for each positive integer $n\ge 4$ we fix the set
$\calF_n$ to be the (finite) set of tournaments on at most $n+1$ vertices that do not embed into
$\bT\bC_n$ (up to isomorphism). }

\begin{remark}\label{rmk:Fn}
    Since every (induced subgraph) tournament on $n-1$ vertices in $\bT\bC_n$ is either $\bT\bT_{n-1}$ or $\bT\bC_{n-1}$,
    the following equality holds
    \[
    \calF_{n-1}\setminus \calF_n = \{\bT\bT_{n-1},\bT\bC_n\}.
    \]
\end{remark}

Building on Lemma~\ref{lem:TC4-min-obs} we prove the following statement.

\begin{lemma}\label{lem:TCn-min-obs}
\blue{For every  $\vec{\bP}_3$-free digraph $\bD$ the following statements are equivalent
\begin{itemize}
    \item $\bD\to \bT\bC_n$, and
    \item $\bD$ is a $\calF_n$-free loopless oriented graph.
\end{itemize}}
\end{lemma}
\begin{proof}
    \blue{The first item clearly implies the second one. We prove the converse implication by induction
    over $n$, and the base case $n = 4$ follows from Lemma~\ref{lem:TC4-min-obs}.
    Let $\bD$ be a $\vec{\bP}_3$-free digraph which is also $\calF_n$-free
    for some positive integer $n \ge 5$, and without loss of generality assume that $\bD$ is weakly
    connected. Clearly, if $\bD$ is also 
    $\calF_{n-1}$-free, then, by induction, we know that $\bD\to \bT\bC_{n-1}$
    and so, $\bD\to  \bT\bC_n$. Now assume that $\bD$ is not $\calF_{n-1}$-free,
    and let $\bF\in \calF_{n-1}$ be a witness of this fact. Since $\bD$ is $\calF_n$-free,
    it follows from Remark~\ref{rmk:Fn} that either $\bF \cong \bT\bT_{n-1}$ or 
    $\bF\cong \bT\bC_n$. We conclude the proof by distinguishing whether $\bD$ contains
    an induced copy of $\bT\bC_n$.
    \begin{itemize}
        \item Suppose that $\bD$ \textbf{does not contain $\bT\bC_n$} as an induced subgraph. 
        In this case $\bD$ contains a transitive tournament on $n-1$ vertices (as witnessed
        by $\bF$), and $\bD$ contains no tournament on $n$ vertices (because $\bD$ is
        $\calF_n$-free and does not contain $\bT\bC_n$).
        We now show that $\bD$ is an acyclic digraph. Proceeding by contradiction
        assume that $\bD$ contains a directed cycle, and using the fact that $\bD$ is $\vec{\bP}_3$-free
        notice that the shortest directed cycle in $\bD$ is a $3$-cycle $c_1c_2c_3$. 
        Let $f_1,\dots, f_{n-1}$ be the vertices of $\bF$ and assume that $(f_i,f_j)\in E(\bD)$ if
        and only if $i < j$. We argue that there is a directed $3$-cycle $\bC$ whose vertex set
        intersects the vertex set of $\bF$. If $\{c_1,c_2,c_3\}\cap \{f_1, \dots, f_{n-1}\}\neq \varnothing$
        there is nothing left to prove. Since $\bD$ is weakly connected, there is an oriented
        $c_1f_i$-path for some $i\in[n-1]$. Consider the shortest oriented $c_1f_i$-path $d_1 = c_1,d_2\dots, d_k = f_i$
        and without loss of generality assume that $d_2 \not\in\{c_2,c_3\}$. We argue by finite
        induction that every vertex of this path belongs to a directed $3$-cycle, and thus, $f_i$ belongs to 
        a directed $3$-cycle. Suppose that there are vertices $u,v\in D$ such that $(d_i,u),(u,v),(v,d_i)\in E(\bD)$.
        If $(d_i,d_{i+1})\in E(\bD)$, then $v = d_{i+1}$ or $v$ and $d_{i+1}$ must be adjacent since $\bD$ is
        $\vec{\bP}_3$-free.  If $v = d_{i+1}$ or $(d_{i+1}, v)\in E(\bD)$, then $d_i,d_{i+1}, v$ witness that
        $d_{i+1}$ belongs to a directed $3$-cycle. Otherwise, $(v,d_{i+1})\in E(\bD)$, and since $\bD$ is
        $\calF_n$-free and $\bT_4^b\in\calF_n$ there is no edge $(u,d_{i+1})$ in $\bD$. This implies
        that either $d_{i+1} = u$ or $(d_{i+1},u)\in E(\bD)$ (because $\bD$ is $\vec{\bP}_3$-free), and
        in both cases $d_{i+1}$ belongs to a directed $3$-cycle. The case when $(d_{i+1},d_i)\in E(\bD)$ follows
        with symmetric arguments. We thus conclude by finite induction that there is some vertex $f_i \in F$
        that belongs to a directed $3$-cycle. With similar arguments as above, one can use such a directed $3$-cycle
        together with $\bF$ and the fact that
        $\bD$ is $\vec{\bP}_3$-free to find a tournament on $n$ vertices in $\bD$. This contradicts
        the assumption that $\bD$ is a $\calF_n$-free oriented graph that does not contain $\bT\bC_n$. 
        Finally, using the observation that $\bD$ is a $\vec{\bP}_3$-free acyclic digraph with no tournament on $n$ vertices
        we conclude that there is no homomorphism $f\colon \vec{\bP}_n\to \bD$. Hence, $\bD\to \bT\bT_{n-1}$
        (Observation~\ref{obs:P-TT}), and thus $\bD\to \bT\bC_n$.
        \item Suppose that $\bD$ \textbf{contains a copy of $\bT\bC_n$}, and let $d_1,\dots, d_n$ be vertices
        such that $i\mapsto d_i$ defines an embedding of $\bT\bC_n$ into $\bD$. Let $X_1:= N^+(d_n)$, $X_n:=N^-(d_1)$,
        and $Y := N^+(d_1)$. With analogous arguments as in the fourth itemized case of the proof of
        Lemma~\ref{lem:TC4-min-obs} it follows that $Y = N^+(d_1) = N^-(d_n)$, that every vertex in $X_1$ (resp.\ in $X_n$)
        has the same in- and out-neighbourhood as $d_1$ (resp.\ as $d_n$), that $(X_1,Y,X_2)$ is a partition
        of the vertex set of $\bD$, and that the subdigraph of $\bD$ induced by $Y$ homomorphically maps to
        $\bT\bT_{n-2}$. Hence, by homomorphically mapping the subdigraph $\bD[Y]$ of $\bD$ to the subdigraph
        $\bT\bC_n[\{2,\dots, n-1\}]$ of $\bD$, all vertices in $X_1$ to $1$, and all vertices in $X_n$ to $n$, 
        we obtain a homomorphism $f\colon \bD\to \bT\bC_n$. 
    \end{itemize}
    }
\end{proof}

Lemma~\ref{lem:TCn-min-obs}  implies that for every positive integer $n$ there is a polynomial-time
algorithm that solves $\CSP(\bT\bC_n)$ when the input is restricted to $\vec{\bP}_3$-free digraphs. 
Hence, the following classification follows from this lemma and from Lemma~\ref{lem:TCkxTTk}.

\begin{theorem}\label{thm:TC-Pk-free-classification}
    For every positive integer $k$ and $n\ge 4$ one of the following holds
    \begin{itemize}
        \item $k \le 3$ and in this case $\CSP(\bT\bC_n)$ is tractable
        where the input is restricted to $\vec{\bP}_k$-free digraphs, or
        \item $k \ge 4$ and in this case $\CSP(\bT\bC_n)$ is  $\NP$-complete
        where the input is restricted to $\vec{\bP}_k$-free digraphs.
    \end{itemize}
\end{theorem}

In the rest of this section we improve Lemma~\ref{lem:TCn-min-obs} by listing the (finitely many)
\emph{minimal} $\vec{\bP}_3$-free obstructions to $\CSP(\bT\bC_n)$. To do so, we introduce two
new tournaments on five vertices depicted in Figure~\ref{fig:tournaments-5-vertices}, and
we prove the following lemma.

\begin{lemma}\label{lem:2-triangles}
    Let $\bD$ be a $\{\bT_4^a, \bT_4^b\}$-free tournament. If $\bD$ contains a pair of directed triangles
    with no common edge, then $\bD$ contains a subtournament $\bD'$ on five vertices that
    also contains a pair of directed triangles with no common edge. 
\end{lemma}
\begin{proof}
    Let $c_1c_2c_3$ and $d_1d_2d_3$ be a pair of directed $3$-cycles in $\bD$ without a common
    edge. If $\{c_1,c_2,c_3\}\cap \{d_1,d_2,d_3\}\neq \varnothing$ there is nothing left to prove,
    so suppose that $|\{c_1,c_2,c_3,d_1,d_2,d_3\}| = 6$. Now notice that since $\bD$ is $\{\bT_4^a,\bT_4^b\}$-free,
    $c_i$ has at least one out-neighbour and one in-neighbour in $\{d_1,d_2,d_3\}$, and symmetrically
    each $d_i$ has at least one out-neighbour and one in-neighbour in $\{c_1,c_2,c_3\}$. Without loss
    of generality assume that $(d_1,c_1),(d_1,c_2),(c_3,d_1)\in E(\bD)$. So, if $(d_2,c_3)\in E(\bD)$,
    then $d_1d_2c_3$ and $c_1c_2c_3$ are a pair of directed $3$-cycles with no common edge and spanning 
    five vertices. Otherwise, $(c_3,d_2)\in E(\bD)$, and since $(c_3,d_1)\in E(\bD)$ and $\bD$
    is $\bT_4^a$-free, it must be the case that $(d_3,c_3)$ is an edge of $\bD$. In this case
    $c_3d_2d_3$ and $c_1c_2c_3$ are a pair of directed $3$-cycles in $\bD$ with no common edge
    and induction a tournament on five vertices. The claim of the lemma follows. 
\end{proof}

\begin{figure}[ht!]
\centering
\begin{tikzpicture}

  \begin{scope}[xshift = -2.5cm, scale = 0.8]
    \node [vertex] (1) at (90:2) {};
    \node [vertex] (2) at (18:2) {};
    \node [vertex] (3) at (306:2) {};
    \node [vertex] (4) at (234:2) {};
    \node [vertex] (5) at (162:2) {};
    \node (L1) at (0,-2.25) {$\bR\bT_5$};
      
    \foreach \from/\to in {1/2, 2/3, 3/4, 4/5, 5/1, 3/1, 1/4, 4/2, 2/5, 5/3}     
    \draw [arc] (\from) to (\to);
  \end{scope}

  \begin{scope}[xshift = 2.5cm, scale = 0.8]
    \node [vertex] (1) at (90:2) {};
    \node [vertex] (2) at (18:2) {};
    \node [vertex] (3) at (306:2) {};
    \node [vertex] (4) at (234:2) {};
    \node [vertex] (5) at (162:2) {};
    \node (L1) at (0,-2.25) {$\bT_5^c$};
      
    \foreach \from/\to in {1/2, 2/3, 3/4, 1/3, 1/4, 2/4, 5/1, 5/2, 3/5, 4/5}     
    \draw [arc] (\from) to (\to);
  \end{scope}

\end{tikzpicture}
\caption{The unique $\{\bT_4^a,\bT_4^b\}$-free tournaments on five vertices that contain a pair
of directed triangles with no a common edge (up to isomorphism).}
\label{fig:tournaments-5-vertices}
\end{figure}
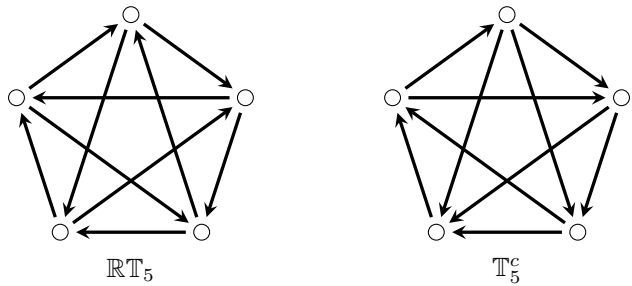

\begin{theorem}\label{thm:TCn-min-obs}
    The following statements are equivalent for 
    every positive integer $n\ge 4$ and every $\vec{\bP}_3$-free digraph
    $\bD$.
    \begin{itemize}
        \item $\bD \to \bT\bC_n$, and
        \item $\bD$ is a $\{\bT_4^a,\bT_4^b, \bR\bT_5, \bT_5^c, \bT\bT_n,\bT\bC_{n+1}\}$-free
        loopless oriented graph.
    \end{itemize}
\end{theorem}
\begin{proof}
    The second itemized statement is clearly a necessary condition for the first itemized statement to 
    hold. \blue{We show that it is also sufficient by proving the contrapositive statement:
    if $\bD$ does not homomorphically map to $\bT\bC_n$, then $\bD$ is not a
    $\{\bT_4^a,\bT_4^b, \bR\bT_5, \bT_5^c, \bT\bT_n,\bT\bC_{n+1}\}$-free loopless oriented graph.
    Suppose that $\bD\not\to \bT\bC_n$, and that $\bD$ is a loopless oriented graph. It follows
    from  Lemma~\ref{lem:TCn-min-obs} that $\bD$ is not $\calF_n$-free --- recall that $\calF_n$
    is the set of tournaments on at most $n+1$ vertices that do not embed into $\bT\bC_n$. Hence,
    it suffices to show that if $\bD$ is a tournament on at most $n+1$ vertices that does not embed
    into $\bT\bC_n$, then $\bD$ contains some tournament from $\{\bT_4^a,\bT_4^b, \bR\bT_5, \bT_5^c,
    \bT\bT_n,\bT\bC_{n+1}\}$.} Since both tournaments on three vertices embed into $\bT\bC_n$ for every
    $n\ge 4$, it follows that $\bD$ has at least $4$ vertices.
    If $\bD$ contains no directed triangle, then $\bD$ is a transitive tournament, and since
    every transitive tournament on at most $n-1$ vertices embeds into $\bT\bC_n$, it must be the
    case that $\bD$ contains a $\bT\bT_n$. Now suppose that $\bD$ contains a directed triangle, 
    and an edge $(u,v)$ such that all directed $3$-cycles of $\bD$ contain the edge $(u,v)$. It
    is not hard to notice that in this case $\bD$ is isomorphic to $\bT\bC_m$ for some positive
    integer $m\ge 4$. Since $\bT\bC_k$ is an induced subtournament of $\bT\bC_n$ whenever $k \le n$, 
    it follows that $m\ge n+1$, and thus $\bD\cong \bT\bC_{n+1}$ (because $|D|\leq n+1$). Finally,
    suppose that $\bD$ contains a pair of directed $3$-cycles with no common edge. Also assume
    that $\bD$ is $\{\bT_4^a,\bT_4^b\}$-free (otherwise there is nothing left to prove). Hence,
    by Lemma~\ref{lem:2-triangles}, $\bD$ contains a subtournament $\bD'$ on five vertices that contains
    two directed triangles with no common edge. The claim now follows because the only two
    such tournament on five vertices which are also $\{\bT_4^a,\bT_4^b\}$-free are $\bR\bT_5$ and
    $\bT_5^c$.
\end{proof}

\section{Omitting a single digraph}
\label{sect:single-digraph}

As mentioned in the introduction, the long-term question of the research line introduced in 
this paper is the following.

\begin{question}\label{qst:long-term}
    Is there a $\cP$ versus $\NP$-complete dichotomy of $\CSP(\bH)$ where the input
    is restricted 
    \begin{enumerate}
        \item to $\bF$-free digraphs?
        \item to $\bF$-subgraph-free digraphs?
    \end{enumerate}
\end{question}

Having settled  Question~\ref{qst:Pk} for the digraphs on three vertices
and for the family of tournaments $\bT\bC_n$, a natural next step
is tackling Question~\ref{qst:long-term} for these digraphs.
Regarding digraphs $\bH$ on three vertices, we leave Question~\ref{qst:long-term} (2) 
open for $\bK_3$ and $\vec{\bC}_3^{++}$, and notice that (1) has a simple solution in
these cases (and $\bF$ connected).

\begin{corollary}
    For every connected digraph $\bF$ the following statements hold.
    \begin{itemize}
        \item Either $\bF \cong \bT\bT_2$, then $\CSP(\vec{\bC}_3^{++})$ is polynomial-time solvable
        for $\bF$-free digraphs, and $\CSP(\bK_3)$ is $\NP$-hard for $\bF$-free digraphs, or
        \item otherwise, $\CSP(\vec{\bC}_3^{++})$ and $\CSP(\bK_3)$ are $\NP$-hard for 
        $\bF$-free digraphs. 
    \end{itemize}
\end{corollary}
\begin{proof}
    The case of $\CSP(\bK_3)$ is straightforward to observe. The first itemized
    claim for $\vec{\bC}_3^{++}$ follows from Theorem~\ref{thm:Pk-3vertex-classification},
    and the second one by Sparse Incomparability for $\bF$ not an oriented forest,
    and by Lemma~\ref{lem:C3++-P3-free} when $\bF$ is an oriented forest. 
\end{proof}

In the remaining of this section we see that some of our proof already yield the first steps 
for settling Question~\ref{qst:long-term} for $\vec{\bC}_3^+$ and the family of tournaments.
The main result in this direction being the following one. 

\begin{theorem}\label{thm:non-paths}
    For every positive integer $n$ and every digraph $\bF$ which is not a disjoint union of oriented paths
    the following statements hold.
    \begin{itemize}
        \item $\CSP(\vec{\bC}_3^+)$ is $\NP$-hard even when the input is restricted to $\bF$-subgraph-free digraphs.
        \item $\CSP(\bT\bC_n)$ is $\NP$-hard even when the input is restricted to $\bF$-subgraph-free digraphs.
    \end{itemize}
\end{theorem}
\begin{proof}
    If $\bF$ contains some vertex $v$ such that $d^+(v) + d^-(v) \ge 3$, then $\bF$ contains an orientation
    of $\bK_{1,3}$, and so the itemized claims follow from Lemma~\ref{lem:C3+-P4-free} and
    Lemma~\ref{lem:TCkxTTk}, respectively. In particular, this proves both claims for oriented forests. 
    If $\bF$ is not an oriented forest, i.e., it contains an oriented cycle, then both itemized statements
    follows from Sparse Incomparability Lemma (Corollary~\ref{thm:sparse-incomparability}).
\end{proof}

The following statement asserts that the CSPs of $\bT\bC_n$ and of $\vec{\bC}_3^+$ remain
$\NP$-hard when restricted $\bP$-subgraph-free digraphs whenever $\bP$ 
is a path that contains two pairs of consecutive edges oriented in the same direction.

\begin{proposition}\label{prop:P3+P3}
    The following statements hold for every connected digraph $\bF$ that contains
    $\vec{\bP}_3+\vec{\bP}_3$ as a subgraph.
    \begin{itemize}
        \item $\CSP(\vec{\bC}_3^+)$ is $\NP$-hard even when the input is restricted to
        $\bF$-subgraph-free digraphs.
        \item$\CSP(\bT\bC_n)$ is $\NP$-hard even when the input is restricted to
        $\bF$-subgraph-free digraphs.
    \end{itemize}
\end{proposition}
\begin{proof}
We make a reduction from 1-IN-3-SAT similar to how we did in the proofs of Lemmas~\ref{lem:C3+-P4-free} and \ref{lem:TCkxTTk}. We use the gadget $\exists w E(w,x) \wedge E(w,y)$, which defines an equivalence relation on both these digraphs, to ensure that the ${\vec \bP_3}$s (which only exist in the clause gadgts) are further apart than they are in $\bF$. We do this by placing $\exists w E(w,x) \wedge E(w,y)$ gadgets end-to-end as they run into the clause gadgets. If the graph underlying $\bF$ is of diameter $d$, then it will suffice to add $d$ copies of this gadget in series at the points at which the variable gadgets meet the clause gadgets.

\end{proof}

Consider a word $w\in\{\leftarrow,\rightarrow\}^\ast$, i.e., a sequence
$w:=w_1\dots w_n$ where $w_i\in\{\leftarrow,\rightarrow\}$ for each $i\in[n]$. We denote
by $\bP_{n+1}^w$ the oriented path with vertex set $[n+1]$ where there is an edge $(i,i+i)$
if $w_i = \rightarrow$, and there is an edge $(i+1,i)$ if $w_i = \leftarrow$. 
In particular, $\vec{\bP}_n = \bP_n^w$ where $w$ is the constant word on $n-1$ letters
$\rightarrow$.

\begin{corollary}\label{cor:P8-alternation}
    The following statements hold for every digraph $\bF$ that contains (as a subgraph)
    the oriented path $\bP_7^{(\leftarrow\rightarrow)^3}$, the path
    $\bP_7^{(\rightarrow\leftarrow)^3}$.
    \begin{itemize}
        \item $\CSP(\vec{\bC}_3^+)$ is $\NP$-hard even when the input is restricted to $\bF$-subgraph-free digraphs.
        \item $\CSP(\bT\bC_n)$ is $\NP$-hard even when the input is restricted to $\bF$-subgraph-free digraphs
        for every $n\ge 4$.
    \end{itemize}
\end{corollary}
\begin{proof}
    We consider the case when $\bF$ contains $\bP_7^{(\leftarrow\rightarrow)^3}$ as a subgraph, 
    en the remaining one follows dually. For (1) (resp.\ for (2)) We use the same proof
    as for Lemma~\ref{lem:C3+-P4-free} (resp.\ for Lemma~\ref{lem:TCkxTTk}) except that we
    change the variable gadget such that it is no longer built from an undirected cycle of
    length $n_v$ but rather identify all vertices that correspond to the same variable. 
\end{proof}

\begin{corollary}\label{cor:bounded-size}
    The following statements hold for every connected digraph $\bF$ on at least $12$ vertices.
    \begin{itemize}
        \item $\CSP(\vec{\bC}_3^+)$ is $\NP$-hard even when the input is restricted to $\bF$-subgraph-free digraphs.
        \item $\CSP(\bT\bC_n)$ is $\NP$-hard even when the input is restricted to $\bF$-subgraph-free digraphs for every $n\ge 4$.
    \end{itemize}
\end{corollary}
\begin{proof}
    By Theorem~\ref{thm:non-paths} and Proposition~\ref{prop:P3+P3}, it suffices to prove the
    claim for  $2\vec{\bP}_3$-subgraph-free paths, and by Lemmas~\ref{lem:C3+-P4-free} 
    and~\ref{lem:TCkxTTk} we also assume that $\bF$ is $\vec{\bP}_5$-subgraph-free. Notice that,
    up to isomorphism, such a path is a subpath of $\bP_{2n+2m+4}^{(\rightarrow\leftarrow)^n\rightarrow\rightarrow\rightarrow(\leftarrow\rightarrow)^m}$, or of
    $\bP_{2n+2m+3}^{(\rightarrow\leftarrow)^n\rightarrow\rightarrow(\leftarrow\rightarrow)^m}$. Finally, 
    any such path on at least 12 vertices contains either $\bP_7^{(\leftarrow\rightarrow)^3}$
    or $\bP_7^{(\rightarrow\leftarrow)^3}$, so we conclude via
    Corollary~\ref{cor:P8-alternation}.
\end{proof}

\subsection*{Forbidden paths on three vertices}

Notice that if $\bD$ is an oriented graph with no directed path on three vertices, 
then $\bD\to \bT\bT_2$ (Observation~\ref{obs:P-TT}). Also, if $\bD$ contains a
symmetric pair of edges $(u,v),(v,u)$, then the subgraph with vertices $u,v$
is a connected component of $\bD$. With these simple arguments one can notice
that for any digraph $\bH$, the problem $\CSP(\bH)$ is in $\cP$ when the
input is restricted to $\vec{\bP}_3$-subgraph-free digraphs.

\begin{observation}\label{obs:general-P3}
    For every digraph $\bH$, $\CSP(\bH)$ is in $\cP$ for $\vec{\bP}_3$-subgraph-free
    digraphs. 
\end{observation}

Hell and Mishra~\cite{HM14} proved that, for any digraph $\bH$, the problem $\CSP(\bH)$ is polynomial-time
solvable where the input is a $\vec{\bP}_3^{\leftarrow\rightarrow}$-subgraph-free or a
$\vec{\bP}_3^{\rightarrow\leftarrow}$-subgraph-free digraphs. 
Here we briefly argue that if $\bH$ is an oriented graph, then $\CSP(\bH)$ is polynomial-time
solvable when the input is a $\vec{\bP}_3^{\leftarrow\rightarrow}$-free digraph. 

A \textit{tree decomposition} for a \blue{digraph} \blue{$\bG=(G,E)$} is a pair \blue{$(\bT,X)$} where $\bT$ is a tree
and $X$ consists of subsets of vertices from \blue{$G$} which we call bags. Each node of \blue{$\bT$} corresponds to a single
bag of $X$. For each vertex $v \in G$ the nodes of $\bT$ containing $v$ must induce a non-empty connected subgraph of
$\bT$ and for each edge $(u,v) \in E$, there must be at least one bag containing both $u$ and $v$. We can then define the \textit{width} of $(\bT,X)$ to be one less than the size of the largest bag. From this, the \textit{treewidth} of a digraph, $\mathit{tw}(G)$, is the minimum width of any \textit{tree decomposition}.

\begin{lemma}
    If $\bH$ is a finite \blue{oriented graph}, then  $\CSP(\bH)$ is in $\cP$ for both the class of $\vec{\bP}_3^{\leftarrow\rightarrow}$-free \blue{digraphs} and the class of
    $\vec{\bP}_3^{\rightarrow\leftarrow}$-free \blue{digraphs}.
    \label{lem:martin-milanic}
\end{lemma}
\begin{proof}
    We make the argument for $\vec{\bP}_3^{\leftarrow\rightarrow}$-free graphs, the other case is dual. 
    \blue{We also assume that the input digraph $\bD$ is an oriented graph (otherwise, we immediately
    reject because $\bH$ is an oriented graph).} Let \blue{$|H|=m$},
    let $\calC_{m+1}$ be the class of digraphs that are $\vec{\bP}_3^{\leftarrow\rightarrow}$-free and further do not contain
    a \blue{tournament} of size $m+1$.
    
    Let $\bG$ be the \blue{symmetric closure of $\bD$, equivalently the} undirected graph underlying $\bD$.
    \blue{Clearly, $\bG$ admits an $\vec{\bP}_3^{\leftarrow\rightarrow}$-free orientation (namely, $\bD$),
    which is exactly the definition of $1$-perfectly orientable graphs from~\cite{HM17}.} According to
    Theorem 6.3 in~\cite{HM17}, this class of graphs is $\bK_{2,3}$-induced-minor-free. If $\mathcal B$ is a graph class
    that is $\bK_{2,3}$-induced-minor-free, then, for each $n$, there is $f(n)$ so that every element of $\mathcal B$
    either contains an $n$-clique or has treewidth bounded by $f(n)$. This is part of Corollary 4.12 in~\cite{DMS21}.
    It follows that $\calC_{m+1}$ has bounded treewidth, and we may plainly assume our input $\bD$ is in $\calC_{m+1}$
    by preprocessing out some no-instances that contain a large oriented clique.
    That $\CSP(\bH)$ can be solved in polynomial time on instances of bounded treewidth is known from~\cite{Dechter}.
\end{proof}

\begin{table}[ht!]
    \begin{center}
    \begin{tabular}{|c|c|c|}
\hline
Constraint (\mbox{e.g.} $\bF:=$)  & $\CSP(\bT\bC_n)$ on $\bF$-subgraph-free & $\CSP(\bT\bC_n)$ on $\bF$-free \\
\hline
$\bF$ is not an oriented tree & \multicolumn{2}{c|}{$\NP$-complete Sparse Incomparability (Corollary~\ref{cor:large-girth})} \\
\hline 
$\bF$ is not an oriented path & \multicolumn{2}{c|}{$\NP$-complete (Theorem~\ref{thm:non-paths})} \\
\hline 
$\bF$ contains $\vec{\bP}_3+\vec{\bP}_3$ as a subgraph & \multicolumn{2}{c|}{$\NP$-complete (Proposition~\ref{prop:P3+P3})} \\
\hline
$|F|\ge 12$ & \multicolumn{2}{c|}{$\NP$-complete (Corollary~\ref{cor:bounded-size})}\\
\hline
$\leftrightarrow$ & $\NP$-complete~\cite{bangjensenSIDMA1} &  $\NP$-complete~\cite{bangjensenSIDMA1}\\
\hline
$\leftarrow \rightarrow$ & P  \cite[Lemma 1]{HM14} & P (Lemma~\ref{lem:martin-milanic}) \\
\hline
$\rightarrow \leftarrow$ & P  \cite[Lemma 1]{HM14} & P (Lemma~\ref{lem:martin-milanic})  \\
\hline
$\rightarrow \rightarrow$ & P (Theorem~\ref{thm:TCn-Pk-subgraph-classification}) & P (Theorem~\ref{thm:TCn-min-obs})\\
\hline
$\rightarrow \rightarrow \rightarrow$ & P (Theorem~\ref{thm:TCn-Pk-subgraph-classification}) & $\NP$-complete (Theorem~\ref{thm:TCn-min-obs})\\
\hline
$\rightarrow \rightarrow \rightarrow \rightarrow $ & $\NP$-complete (Theorem~\ref{thm:TCn-Pk-subgraph-classification}) & $\NP$-complete  (Theorem~\ref{thm:TCn-min-obs})\\
\hline
$\rightarrow \rightarrow \leftarrow\leftarrow$ & \multicolumn{2}{c|}{$\NP$-complete (Lemma~\ref{lem:TCkxTTk})}\\
\hline
$\leftarrow\leftarrow \rightarrow\rightarrow$ & \multicolumn{2}{c|}{$\NP$-complete (Lemma~\ref{lem:TCkxTTk})} \\
\hline
$(\leftarrow \rightarrow)^3$ & \multicolumn{2}{c|}{$\NP$-complete (Corollary~\ref{cor:P8-alternation})} \\
\hline
$(\rightarrow \leftarrow)^3$ & \multicolumn{2}{c|}{$\NP$-complete (Corollary~\ref{cor:P8-alternation})}\\
\hline
    \end{tabular}
    \end{center}
    \caption{Complexity landscape for $\CSP(\bT\bC_n)$ under the omission of  single connected subgraph or induced connected subgraph.}
    \label{fig:landscape-TCn-omitting-subgraph-bis}
\end{table}

\begin{table}[ht!]
    \begin{center}
    \begin{tabular}{|c|c|c|}
\hline
Constraint (\mbox{e.g.} $\bF:=$)  & $\CSP(\vec{\bC}_3^+)$ on $\bF$-subgraph-free & $\CSP(\vec{\bC}_3^+)$ on $\bF$-free \\
\hline
$\bF$ is not an oriented tree & \multicolumn{2}{c|}{$\NP$-complete Sparse Incomparability (Corollary~\ref{cor:large-girth})} \\
 \hline 
$\bF$ is not an oriented path & \multicolumn{2}{c|}{$\NP$-complete (Theorem~\ref{thm:non-paths})} \\
\hline 
$\bF$ contains $\vec{\bP}_3+\vec{\bP}_3$ as a subgraph & \multicolumn{2}{c|}{$\NP$-complete (Proposition~\ref{prop:P3+P3})} \\
\hline
$|F|\ge 12$ & \multicolumn{2}{c|}{$\NP$-complete (Corollary~\ref{cor:bounded-size})} \\
\hline
$\leftarrow \rightarrow$ & P  \cite[Lemma 1]{HM14} & \blue{Open} \\
\hline
$\rightarrow \leftarrow$ & P  \cite[Lemma 1]{HM14} & \blue{Open} \\
\hline
$\rightarrow \rightarrow$ & P (Theorem~\ref{thm:3-vertices-Pk-subgraph-free}) & P (Theorem~\ref{thm:Pk-3vertex-classification})\\
\hline
$\rightarrow \rightarrow \rightarrow$ & P (Theorem~\ref{thm:3-vertices-Pk-subgraph-free}) & $\NP$-complete (Theorem~\ref{thm:Pk-3vertex-classification})\\
\hline
$\rightarrow \rightarrow \rightarrow \rightarrow $ & $\NP$-complete (Theorem~\ref{thm:3-vertices-Pk-subgraph-free}) & $\NP$-complete  (Theorem~\ref{thm:Pk-3vertex-classification})\\
\hline
$\rightarrow \rightarrow \leftarrow\leftarrow$ & \multicolumn{2}{c|}{$\NP$-complete (Lemma~\ref{lem:C3+-P4-free})} \\
\hline
$\leftarrow\leftarrow \rightarrow\rightarrow$ & \multicolumn{2}{c|}{$\NP$-complete (Lemma~\ref{lem:C3+-P4-free})}\\
\hline
$(\leftarrow \rightarrow)^3$ & \multicolumn{2}{c|}{$\NP$-complete (Corollary~\ref{cor:P8-alternation})}\\ 
\hline
$(\leftarrow \rightarrow)^3$ & \multicolumn{2}{c|}{$\NP$-complete (Corollary~\ref{cor:P8-alternation})} \\
\hline
    \end{tabular}
    \end{center}
    \caption{Complexity landscape for $\CSP(\vec{\bC}_3^+)$ under the omission of  single connected subgraph or induced connected subgraph.}
    \label{fig:landscape-C3+-omitting-subgraph}
\end{table}

\section{Conclusion and outlook}

In this paper we have brought together homomorphisms, digraphs and $\bH$-(subgraph-)free algorithmics. 
In doing so, we have uncovered a series of results concerning not only restricted CSPs, but also 
hardness of digraph CSPs under natural restrictions such as acyclicity. Our work raises numerous open 
problems, which we believe deserve attention in the future.
Besides Questions~\ref{qst:Pk} and~\ref{qst:long-term}, we ask the following.

\begin{itemize}
    \item Is it true that for every finite structure $\bA$ and every (possibly infinite) $\bB$,
    if $(\bA,\bB)$ does not rpp-construct $(\bK_3,\bL)$, then $\RCSP(\bA,\bB)$ is
    polynomial-time solvable? (Compare to Theorem~\ref{thm:finite-RCSP-dichotomy}).
    \item Is it true that for every finite structure $\bA$ and every (possibly infinite) $\bB$
    the problem $\RCSP(\bA,\bB)$ is either in $\cP$ or $\NP$-hard?
    (Compare to Theorem~\ref{thm:finite-RCSP-dichotomy}).
    \item Let $\bA$ be a finite structure and $\bB$ a structure whose $\CSP$ is in $\GMSNP$. 
    Is it true that if $(\bA,\bB)$ does not rpp-construct $(\bK_3,\bL)$, then $\RCSP(\bA,\bB)$
    is polynomial-time solvable? (Compare to Theorem~\ref{lem:GMSNP-restrictions}).
    \item  Since $\CSP(\bH)$ reduces to $\CSP(\bH)$ restricted to acyclic digraphs
    (Theorem~\ref{thm:acyclic+bounded-paths}),  and the latter is polynomial-time equivalent
    to $\CSP(\bH\times \bQ)$ we ask: is it true that for every finite digraph $\bH$ the (infinite)
    digraph $\bQ\times \bH$ pp-constructs $\bH$? 
    \item Is is true that for every oriented graph $\bH$ there a finitely many 
    $\vec{\bP}_3$-free minimal obstructions to $\CSP(\bH)$? (Compare to Theorem~\ref{thm:TCn-min-obs}).
    \item Is is true that for every tournament $\bT$ there a finitely many 
    $\vec{\bP}_3$-free minimal obstructions to $\CSP(\bT)$? (Compare to Theorem~\ref{thm:TCn-min-obs}).
    \item Settle Question~\ref{qst:long-term} for digraphs on three vertices (see also Table~\ref{fig:landscape-C3+-omitting-subgraph}).
    \item Settle Question~\ref{qst:long-term} for tournaments, in particular,
    for the family of tournaments $\bT\bC_n$ (see also Table~\ref{fig:landscape-TCn-omitting-subgraph-bis}).
\end{itemize}

\subsection*{Persistent structures}

A natural question arising from restricted CSPs is if there are structures $\bA$ such that
$\RCSP(\bA,\bB)$ is $\NP$-hard whenever $\bB\not\to \bA$. We say that a structure $\bA$
\emph{persistently constructs} $\bK_3$ if for every (possibly infinite) $\bB$ such
that $\bB\not\to\bA$ the restricted CSP template $(\bA,\bB)$ rpp-constructs $(\bK_3,\bL)$.
Notice that in this case $\RCSP(\bA,\bB)$ is $\NP$-hard whenever $\bB\not\to \bA$.

\begin{observation}\label{obs:persistent}
    For a finite structure $\bA$ the following statements are equivalent.
    \begin{itemize}
        \item $(\bA,\bB)$ rpp-constructs $(\bK_3,\bL)$ for every finite structure $\bB\not\to \bA$.
        \item $(\bA,\bB)$ rpp-constructs $(\bK_3,\bL)$ for every (possibly infinite) structure $\bB\not\to \bA$.
        \item $\RCSP(\bA,\bB)$ is $\NP$-hard for every finite structure $\bB\to \bA$ (assuming $\cP\neq \NP$).
    \end{itemize}
\end{observation}
\begin{proof}
    The equivalence between the first and third statement follows from the assumption that $\cP\neq \NP$,
    and from Theorem~\ref{thm:finite-RCSP-dichotomy}. The second item clearly implies the first one. Finally,
    to see that the first one implies the second one, suppose $\bB\not\to \bA$. By compactness
    there is a finite substructure $\bB'$ of $\bB$ such that $\bB'\not\to\bA$. Hence, $(\bA,\bB')$
    rpp-constructs $(\bK_3,\bL)$, and clearly $(\bA,\bB)$ rpp-constructs $(\bA,\bB')$. Since rpp-constructions
    compose, we conclude that $(\bA,\bB)$ rpp-constructs $(\bK_3,\bL)$. 
\end{proof}

\begin{remark}
    A similar compactness as in the proof of Observation~\ref{obs:persistent} applies to
    \emph{$\omega$-categorical structures} (see, e.g.,~\cite[Lemma 4.1.7]{Book}), i.e.,
    to structures $\bA$ whose automorphism group has finitely many orbits of $k$-tuples for
    each positive integer $k$. Hence, if $\bA$ is an $\omega$-categorical structure,
    then $\bA$ persistently constructs $\bK_3$ if and only if $(\bA,\bB)$ rpp-constructs
    $(\bK_3,\bL)$ for every finite structure $\bB\not\to \bA$.
\end{remark}

\begin{problem}
    Characterize the class of finite digraphs (structures) that persistently construct $\bK_3$.
\end{problem}

Theorem~\ref{thm:brewster} asserts that for every hereditarily hard digraph $\RCSP(\bH,\bH')$
is $\NP$-hard whenever $\bH'$ is a finite digraph and $\bH'\not\to\bH$. Hence, assuming
$\cP\neq \NP$, this implies that every hereditarily hard digraph persistently constructs
$\bK_3$ (Observation~\ref{obs:persistent}).

\begin{problem}
    Characterize the class of finite digraphs (structures) $\bH$ such that $\RCSP(\bH,\bH')$
    is $\NP$-hard whenever $\bH'\not\to \bH$ (assuming $\cP\neq \NP$).
\end{problem}

\section*{Acknowledgments}
We are grateful to Martin Milani\v{c} for the proof of Lemma~\ref{lem:martin-milanic}.

\bibliographystyle{abbrv}
\bibliography{global.bib}

\end{document}